\documentclass[runningheads]{llncs}

\usepackage[utf8]{inputenc}
\usepackage{amsmath,amssymb,latexsym,stmaryrd}
\usepackage{graphicx} 
\usepackage[colorlinks=true, allcolors=blue]{hyperref}
\usepackage{xcolor} 
\usepackage{xspace} 
\usepackage{xargs} 
\usepackage{ebproof} 
\usepackage{enumitem} 
\usepackage{mathtools} 
\usepackage{mdframed} 
\usepackage{cleveref}
\usepackage{mathdots} 
\usepackage{float} 
\usepackage{environ}
\usepackage{refcount} 
\usepackage{enumitem} 
\usepackage{marvosym}
\usepackage{fancyhdr}

\spnewtheorem{ourdefinition}[definition]{Definition}{\bfseries}{}

\newcounter{theoremRem}

\fancypagestyle{firstpage}{%
  \fancyhf{} 
  \fancyfoot[L]{ \small\textit{Accepted at TABLEAUX 2025} }
}

\newcommand{\calA}{\mathcal{A}} 
\newcommand{\calB}{\mathcal{B}}
\newcommand{\calC}{\mathcal{C}}

\newcommand{\calE}{\mathcal{E}}
\newcommand{\calF}{\mathcal{F}}
\newcommand{\calG}{\mathcal{G}}

\newcommand{\calL}{\mathcal{L}}
\newcommand{\calM}{\mathcal{M}}

\newcommand{\calP}{\mathcal{P}}

\newcommand{\calS}{\mathcal{S}}
\newcommand{\calT}{\mathcal{T}}

\newcommand{\bbS}{\mathbb{S}}

\newcommand{\bbG}{\mathbb{G}}

\newcommand{\sfQ}{\mathsf{Q}}
\newcommand{\sfR}{\mathsf{R}}

\newcommand{\K}[1]{\mathsf{K}_{<#1}}
\newcommand{\cycs}[1]{\mathsf{L}_{<#1}}

\renewcommand{\|}{\! \mid \!}

\renewcommand{\phi}{\varphi}
\renewcommand{\epsilon}{\varepsilon}

\newcommand{\mybar}[1]{\overline{#1}}
\newcommand{\Nat}{\mathbb{N}}

\newcommand{\powset}{\mathcal{P}}

\newcommand{\isdef}{\mathrel{:=}}
\newcommand{\nada}{\varnothing}

\newcommand{\impl}{\rightarrow} 
\newcommand{\liff}{\leftrightarrow} 
\newcommand{\Land}{\bigwedge}

\renewcommand{\land}{\wedge}
\renewcommand{\lor}{\vee}
\newcommand{\proves}{\vdash}
   \newcommand{\cproves}{\proves^c}

\newcommand{\Prop}{\mathsf{Prop}}

\newcommand{\Act}{\mathsf{Act}}
\newcommand{\Voc}{\mathsf{Voc}}
\newcommand{\conv}[1]{\breve{#1}}

\newcommand{\sat}{\Vdash}
\newcommand{\muML}{\calL_{\mu}}

\newcommand{\FL}{\mathsf{FL}}
\newcommand{\FLN}{\FL^{\neg}}

\newcommand{\lbox}{\scalebox{0.8}{$\square$}}

\newcommand{\lm}{\mbox{\ooalign{\ld \cr \hidewidth\raise.05ex\hbox{$* \mkern3.1mu$}\cr}}} 
\newcommand{\lbm}{\ooalign{$\lbox$ \cr \hidewidth\raise.05ex\hbox{$* \mkern5.5mu$}\cr}\hspace{-0.1cm}} \newcommand{\lcm}{\ooalign{$\lbox$ \cr \hidewidth\raise.05ex\hbox{$\cdot \mkern2.9mu$}\cr}}
\newcommand{\ldiap}[1][\alpha]{\ensuremath{\langle#1\rangle}\xspace}
\newcommand{\lboxp}[1][\alpha]{\ensuremath{[#1]}\xspace}
\newcommand{\edia}[2]{\lozenge(#1,#2)}

\newcommand{\PDL}{\ensuremath{\mathsf{PDL}}\xspace}
\newcommand{\CPDL}{\ensuremath{\mathsf{CPDL}}\xspace}

\newcommand{\CPDLf}{\ensuremath{\mathsf{CPDL}_f}\xspace}
\newcommand{\SCPDLf}{\ensuremath{\mathsf{SCPDL}_f}\xspace}
\newcommand{\SCPDLInfty}{\ensuremath{\SCPDLf^\infty}\xspace}
\newcommand{\CPDLInfty}{\ensuremath{\mathsf{CPDL}_f^\infty}\xspace}

\newcommand{\edge}{\lessdot} 
\newcommand{\edgeT}{<} 
\newcommand{\edgeRT}{\leq} 

\newcommand{\cEdge}{\lhd} 
\newcommand{\cEdgeT}{\lhd^{+}} 
\newcommand{\cEdgeRT}{\lhd^{*}} 

\newcommand{\qedge}{\lessdot_{\sfQ}} 

\newcommand{\Plab}{\Lambda} 
\newcommand{\Qlab}{\Psi} 

\newcommand{\cEquiv}{\equiv_c}
\newcommand{\FRC}{\calF_{C}}
\newcommand{\FRCp}{\calF_{C^+}}

\newcommand{\Ru}{\ensuremath{\mathsf{R}}\xspace}
\newcommand{\AxLit}{\ensuremath{\mathsf{Ax1}}\xspace}
\newcommand{\AxBot}{\ensuremath{\mathsf{Ax2}}\xspace}
\newcommand{\RuOr}{\ensuremath{{\lor}}\xspace}
\newcommand{\RuAnd}{\ensuremath{{\land}}\xspace}

\newcommand{\RuDia}[1][\ensuremath{a}]{\ensuremath{\mathsf{{\ldiap[#1]}}}\xspace}

\newcommand{\RuWeak}{\ensuremath{\mathsf{weak}}\xspace}
\newcommand{\RuU}{\ensuremath{\mathsf{u}}\xspace}
\newcommand{\RuF}{\ensuremath{\mathsf{f}}\xspace}
\newcommand{\RuCut}{\ensuremath{\mathsf{cut}}\xspace}
\newcommand{\RuACut}{\ensuremath{\mathsf{acut}}\xspace}

\newcommand{\RuDiaP}{\ensuremath{\mathsf{{\ldiap[a]}}}\xspace}
\newcommand{\RuConB}{\ensuremath{\mathsf{{\lboxp[;]}}}\xspace}
\newcommand{\RuConD}{\ensuremath{\mathsf{{\ldiap[;]}}}\xspace}
\newcommand{\RuChoiceD}{\ensuremath{\mathsf{{\ldiap[\cup]}}}\xspace}
\newcommand{\RuChoiceB}{\ensuremath{\mathsf{{\lboxp[\cup]}}}\xspace}
\newcommand{\RuStarD}{\ensuremath{\mathsf{{\ldiap[*]}}}\xspace}
\newcommand{\RuStarB}{\ensuremath{\mathsf{{\lboxp[*]}}}\xspace}
\newcommand{\RuTestD}{\ensuremath{\mathsf{{\ldiap[?]}}}\xspace}
\newcommand{\RuTestB}{\ensuremath{\mathsf{{\lboxp[?]}}}\xspace}

\newcommand{\cyclicPT}{\calT_{\pi}^C}
\newcommand{\tracestep}{\to_{C}}
\newcommand{\trace}{\twoheadrightarrow_{C}}
\newcommand{\equic}{\equiv_C}

\newcommand{\anColor}[1]{\textcolor{black}{#1}}

\newcommand{\EG}{\mathcal{E}}
\newcommand{\Own}{O}
\newcommand{\WC}{W}
\newcommand{\Om}{\Omega}

\newcommand{\last}{\mathsf{last}}
\newcommand{\first}{\mathsf{first}}
\newcommand{\PM}{\mathit{PM}}

\newcommand{\eloi}{\exists}
\newcommand{\abel}{\forall}
\newcommand{\Win}{\mathit{Win}}

\newcommand{\Inf}{\mathsf{Inf}}

\newcommand{\Seq}{\mathsf{Seq}}
\newcommand{\Inst}{\mathsf{Inst}}
\newcommand{\conc}{\mathsf{conc}}

\newcommand{\tom}[1]{\overset{#1}{\to}}

\newcommand{\mng}[1]{\llbracket #1 \rrbracket}
    \newcommand{\mngS}[1]{\llbracket #1 \rrbracket^{\bbS}}
\newcommand{\gmng}[1]{\llparenthesis #1 \rrparenthesis}
    \newcommand{\gmngS}[1]{\llparenthesis #1 \rrparenthesis^{\bbS}}

\newcommand{\exop}[1]{\langle #1 \rangle}
\newcommand{\unop}[1]{[ #1 ]}

\newenvironment{tbs}{%
\begin{itemize}\tt }{\end{itemize}}
\newcommand{\btbs}{\begin{tbs}}
\newcommand{\etbs}{\end{tbs}}

\title{Interpolation for Converse PDL}
\author{Johannes Kloibhofer\inst{1}\thanks{%
   The research of this author has been made possible by a grant from the 
   Dutch Research Council NWO, project nr. 617.001.857.
   }
\and Valentina Trucco Dalmas\inst{2}{}
\and Yde Venema\inst{1}}
\institute{ILLC, University of Amsterdam, Netherlands \\
\email{{\{j.kloibhofer, y.venema\}}@uva.nl}
\and University of Groningen, Netherlands \\
\email{f.c.trucco.dalmas@rug.nl}}

\authorrunning{J. Kloibhofer et al.}

\begin{document}
\maketitle
\thispagestyle{firstpage}

\begin{abstract}
Converse \PDL is the extension of propositional dynamic logic with a converse operation on programs.
Our main result states that Converse \PDL enjoys the (local) Craig Interpolation Property, with respect to both atomic programs and propositional variables.
As a corollary we establish the Beth Definability Property for the logic.

Our interpolation proof is based on an adaptation of Maehara's proof-theoretic method.
For this purpose we introduce a sound and complete cyclic sequent system for this logic.
This calculus features an analytic cut rule and uses a focus mechanism for recognising successful cycles.

\end{abstract}


\section{Introduction}

Propositional Dynamic Logic (abbreviated: \PDL) was introduced by Fischer \& Ladner~\cite{FL:PDL1979} in 1979 as a propositional formalism to reason about the behaviour of programs.
The language of \PDL features an infinite collection of modalities, the intended interpretation of $\ldiap\phi$ being that `after some execution of the program $\alpha$, the formula $\phi$ holds'.
The inductive structure of programs is reflected by the syntax of \PDL, where complex programs are constructed from atomic ones and formula tests, by means of program constructors for sequential composition, nondeterministic choice and iteration.

Converse \PDL or \CPDL, also defined in~\cite{FL:PDL1979}, extends \PDL with a converse operator on programs, which facilitates backwards reasoning about programs.
\PDL and \CPDL also have applications in for instance knowledge representation~\cite{baad:desc07}, where the program expressions represent \emph{roles} between objects, and the program constructions correspond to natural operations on such roles; in particular, the converse operator corresponds to \emph{inverse roles}.

\PDL and \CPDL both have the small-model property and an \textsc{Exptime}-complete satisfiability problem, as established by Fischer \& Ladner~\cite{FL:PDL1979} and Pratt~\cite{prat:near80}).
A natural axiomatisation was given by Segerberg~\cite{sege:comp77} and proved to be complete by Parikh~\cite{pari:comp78} and others.
Generally, \PDL and related formalisms have been recognized as important modal logics for quite some time now, see for instance Troquard \& Balbiani~\cite{troq:prop23} for a recent survey.

Here we study interpolation properties for \CPDL.
A logic has the \emph{Craig Interpolation Property} (CIP) if any pair $\phi,\psi$ of formulas such that $\phi \models \psi$ has an \emph{interpolant}, that is, a formula $\theta$ in the common vocabulary of $\phi$ and $\psi$ such that $\phi \models \theta$ and $\theta \models \psi$.
(Here $\models$ denotes the local consequence relation, that is: we have $\phi \models \psi$ iff the implication $\phi\to\psi$ is valid.)
A related property is the Beth Definability Property, which informally states that any concept which can implicitly be defined in the logic, in fact also has an explicit definition.
Interpolation and definability are generally considered to be attractive properties of a formalism; applications in computer science include the modularisation and optimisation of reasoning~\cite{mcmi:appl05,rena:inte08,cate:beth13,jung:sepa21}.
Both properties have been studied extensively in the literature on modal logic, see Gabbay \& Maksimova~\cite{gabb:inte05} for a survey.

Our main result states that \CPDL has the Craig Interpolation Property indeed, with respect to both atomic programs and propositional variables.
As a direct corollary we obtain the Beth Definability Property for \CPDL. 

We prove Craig interpolation following Maehara's proof-theoretic method, adapted to our logic. 
That is, we first introduce a sound and complete cyclic proof system for this logic.
This calculus features an analytic cut rule and uses a focus mechanism for recognising successful cycles.
To prove interpolation we consider a version of the proof system that operates on \emph{split sequents}, that is, pairs of finite sets of formulas.
Given a pair $\phi,\psi$ of formulas such that $\phi\models\psi$, we fix a derivation $\pi$ of the split sequent $\mybar{\psi}\|\phi$.
This derivation induces a system of equations, whose solution would immediately yield an interpolant of $\phi$ and $\psi$.
As in the case for standard \PDL, a straightforward approach would result in a system of equations that cannot be solved \emph{inside} \CPDL itself.
To solve this problem we employ an idea from Borzechowski~\cite{borz:tabl88}:
based on an \emph{auxiliary structure} rather than on $\pi$ itself, we set up 
an alternative system of equations, which \emph{does} have a solution in \CPDL, 
thus providing us with the desired interpolant.
\smallskip 

\noindent
\textit{Related work}
Maehara's method has been extended to cyclic proof systems by Shamkanov~\cite{Shamkanov2014}, Afshari \& Leigh~\cite{Afshari2019} and Marti \& Venema~\cite{mart:focu21}, in order to prove interpolation properties for, respectively, G\"odel-L\"ob logic, the modal $\mu$-calculus $\muML$ and its alternation-free fragment.
Closely related to our work is the paper by Borzechowski et al.~\cite{borz:prop25}, who use ideas from Borzechowski~\cite{borz:tabl88} to establish interpolation for \PDL.
Our work extends~\cite{borz:prop25} to include converse modalities: notable differences are that our proof system features a cut rule and a more involved modal rule, while on the other hand our rules for the program constructors are simpler than those in~\cite{borz:prop25}.
Kloibhofer \& Venema establish Craig interpolation for the two-way $\mu$-calculus~\cite{kloi:inte25}.
Our cyclic proof system is much simpler, since it avoids the use of trace atoms and can restrict to a simple focus annotation system; the construction of our interpolants is rather more intricate, however.


\section{Converse PDL}
\label{s:cpdl}

\subsection{Syntax}

\begin{ourdefinition}
Let $\Prop$ be an infinite set of atomic propositions and $\Act$ an infinite set of atomic programs. 
We assume an involution operation $\conv{\cdot}: \Act \to \Act$ such that for every $a \in \Act$ it holds that $a \neq \conv{a}$ and $a = \conv{\conv{a}}$. The sets of \emph{formulas} and \emph{programs} of $\CPDL$ are given by the following mutual induction:
\begin{align*}
	\phi &::= \top \| \bot \| p \| \mybar{p} \| \phi \land \phi \| \phi \lor \phi \| \ldiap \phi \|  \lboxp \phi \\
		\alpha &::= a \| \alpha;\alpha \| \alpha \cup \alpha \| \alpha^* \| \phi?
\end{align*}
where $p \in \Prop$ and $a \in \Act$\footnote{%
   In our set-up the atomic actions come in pairs $a, \conv{a}$. This has some technical advantages over the approach where the converse operation $\conv{}\,$ is an explicit syntactic symbol.
   } 
.
\end{ourdefinition}

The \emph{vocabulary of $\phi$}, written $\Voc(\phi)$, is the set of propositions and actions occurring in $\phi$, with the proviso that we include both $a$ and $\conv{a}$ in the vocabulary of any expression in which $a$ or $\conv{a}$ occurs.
We refer to formulas of the form $\ldiap[\alpha]\phi$ and $\lboxp[\alpha]\phi$ as, respectively, \emph{diamond} and \emph{box formulas}.
Formulas of the form $\ldiap[\alpha^*]\phi$ or $\lboxp[\alpha^*]\phi$ are called \emph{fixpoint formulas}.

We think of the formula $\mybar{p}$ as the negation of $p$, and inductively extend the map $p \mapsto \mybar{p}$ to a full-blown negation operation as follows:
\[\begin{array}{lllclllclllclll}
   \mybar{\top} & \isdef & \bot
  && &&
  && \mybar{\phi \land \psi} & \isdef & \mybar{\phi} \lor \mybar{\psi}
  && \mybar{\lboxp \phi} & \isdef & \ldiap \mybar{\phi} 
\\ \mybar{\bot} & \isdef & \top \qquad
  && \mybar{\mybar{p}} & \isdef & p  \qquad
  && \mybar{\phi \lor \psi} & \isdef & \mybar{\phi} \land \mybar{\psi} \qquad
  && \mybar{\ldiap \phi} & \isdef & \lboxp \mybar{\phi} 
\end{array}\]
For a set of formulas $\Gamma$ we define $\mybar{\Gamma}= \{\mybar{\phi} \| \phi \in \Gamma\}$.
Note that $\mybar{\mybar{\phi}} = \phi$ for every formula $\phi$. 

Similarly, we extend the converse operator to arbitrary programs by putting 
$(\alpha;\beta)\conv{\,} \isdef \conv{\beta};\conv{\alpha}$, $(\alpha\cup\beta)\conv{\,} \isdef \conv{\alpha}\cup\conv{\beta}$,
$(\alpha^{*})\conv{\,} \isdef \conv{\alpha}^{*}$ and 
$(\tau?)\conv{\,} \isdef \tau?$.
Thus we may think indeed of \CPDL as extending \PDL with a converse operation on programs.
   
We let $\phi[\psi/p]$ denote the result of substituting all occurrences of $p$ in the formula $\phi$ with the formula $\psi$
(and all occurrences of $\mybar{p}$ with $\mybar{\psi}$).

\begin{ourdefinition}
Let $\phi$ and $\psi$ be formulas. We write $\phi \tracestep \psi$ if 
\begin{enumerate}
\item $\phi = \chi_0 \land \chi_1$ and $\psi \in \{\chi_0,\chi_1\}$ or $\phi= \chi_0 \lor \chi_1$ and $\psi \in \{\chi_0,\chi_1\}$;
\item $\phi = \ldiap[a] \chi$ and $\psi = \chi$, or
$\phi = \lboxp[a] \chi$ and $\psi = \chi$;
\item $\phi = \ldiap[\alpha;\beta]\chi$ and $\psi = \ldiap[\alpha] \ldiap[\beta]\chi$, 
or $\phi = \lboxp[\alpha;\beta]\chi$ and $\psi = \lboxp[\alpha] \lboxp[\beta]\chi$;
\item $\phi = \ldiap[\alpha\cup\beta]\chi$ and $\psi \in \{ \ldiap[\alpha]\chi, \ldiap[\beta]\chi \}$,
or $\phi = \lboxp[\alpha\cup\beta]\chi$ and $\psi \in \{ \lboxp[\alpha]\chi, \lboxp[\beta]\chi \}$;
\item $\phi = \ldiap[\tau?]\chi$ and $\psi \in \{ \tau, \chi\}$,
or $\phi = \lboxp[\tau?]\chi$ and $\psi \in \{ \mybar{\tau}, \chi \}$;
\item $\phi = \ldiap[\alpha^*] \chi$ and $\psi \in \{\ldiap \ldiap[\alpha^*] \chi, \chi \} $,
or $\phi = \lboxp[\alpha^*] \chi$ and $\psi \in \{ \lboxp \lboxp[\alpha^*] \chi, \chi\}$.
\end{enumerate}
	
A \emph{trace} is a sequence $(\phi_n)_{0\leq n<\kappa}$ (with $\kappa \leq \omega$) such that $\phi_{n}\tracestep\phi_{n+1}$ for all $i<\kappa$.
We define the \emph{trace relation} $\trace$ on formulas as the reflexive-transitive closure of $\tracestep$.
We write $\phi \equic \psi$ if both $\phi \trace \psi$ and $\psi \trace \phi$.
We define the \emph{Fischer-Ladner closure} of a sequent $\Gamma$ as the set $\FL(\Gamma) \isdef \{ \psi \mid \phi \trace \psi \text{ for some }\phi \in \Gamma\}$ and put $\FLN(\Gamma) \isdef \FL(\Gamma) \cup \FL(\mybar{\Gamma})$. 
\end{ourdefinition}

\begin{proposition}
\label{p:inftr}
Let $t = (\phi_{n})_{n<\omega}$ be an infinite trace.
Then infinitely many $\phi_{n}$ are fixpoint formulas, and either cofinitely many $\phi_{n}$ are diamond formulas, or cofinitely many $\phi_{n}$ are box formulas.
\end{proposition}

\subsection{Semantics}

\CPDL-formulas are interpreted in (poly-modal) Kripke models.
It will be convenient for us to present the semantics in terms of an evaluation game
; see Appendix \ref{app.parityGames} for a definition of infinite games. 
This game-theoretic semantics is equivalent to the standard, compositional one; we refer to Appendix \ref{app.CPDL} for a definition of the compositional semantics and a proof of its equivalence. 

\begin{ourdefinition}\label{d.KripkeModels}
A \emph{Kripke model} $\bbS = (S,R, V)$ consists of a set $S$ of states; a family of binary  relations $R = \{R_a \subseteq S^2 \mid a \in \Act\}$ on $S$ such that $(s,s') \in R_a$ iff $(s',s) \in R_{\conv{a}}$; and a valuation $V: \Prop \to \powset(S)$. A \emph{pointed model} is a pair $(\bbS,s)$ where $\bbS$ is a Kripke model and $s \in S$.
\end{ourdefinition}

Let $\bbS = (S,R,V)$ be a Kripke model. The \emph{evaluation game} $\EG(\bbS)$ 
is the following infinite two-player game.
Its positions are pairs of the form $(\phi,s)$, where $\phi$ is a formula and $s\in  S$, and its ownership function and admissible moves are given in the following table.

\vspace{-10pt}
\begin{table}[htb]
	\begin{center}
		\begin{tabular}{|ll|c|l||ll|c|l|}
			\hline
			Position && Owner & Admissible moves & Position && Owner & Admissible moves \\
			\hline
			\multicolumn{2}{|l|}{$(\phi \lor \psi,s)$}   & $\eloi$   
			& $\{ (\phi,s), (\psi,s) \}$ 
			& \multicolumn{2}{|l|}{$(\phi \land \psi,s)$} & $\abel$ 
			& $\{ (\phi,s), (\psi,s) \}$
			\\ \multicolumn{2}{|l|}{$(\ldiap[a] \phi,s)$}        & $\eloi$ 
			& $\{ (\phi,t) \mid (s,t) \in R_a \}$ 
			& \multicolumn{2}{|l|}{$(\lboxp[a] \phi,s) $}       & $\abel$ 
			& $\{ (\phi,t) \mid (s,t) \in R_a\}$ 
			\\ \multicolumn{2}{|l|}{$(\ldiap[\alpha;\beta] \phi,s)$}        & -
			& $\{(\ldiap[\alpha]\ldiap[\beta] \phi,s)  \}$ 
			& \multicolumn{2}{|l|}{$(\lboxp[\alpha;\beta] \phi,s) $}       & -
			& $\{ (\lboxp[\alpha]\lboxp[\beta]\phi,s) \}$ 
			\\ \multicolumn{2}{|l|}{$(\ldiap[\alpha \cup \beta] \phi,s)$}        & $\eloi$ 
			& $\{ (\ldiap[\alpha] \phi,s), (\ldiap[\beta]\phi,s)\}$ 
			& \multicolumn{2}{|l|}{$(\lboxp[\alpha \cup \beta] \phi,s)$}        & $\abel$ 
			& $\{ (\lboxp[\alpha] \phi,s), (\lboxp[\beta]\phi,s)\}$ 	 
			\\ \multicolumn{2}{|l|}{$(\ldiap[\alpha^*] \phi,s)$}        & $\eloi$ 
			& $\{ (\ldiap[\alpha]\ldiap[\alpha^*] \phi,s), (\phi,s) \}$ 
			& \multicolumn{2}{|l|}{$(\lboxp[\alpha^*] \phi,s)$}        & $\abel$ 
			& $\{ (\lboxp[\alpha]\lboxp[\alpha^*] \phi,s), (\phi,s) \}$ 			
			\\ \multicolumn{2}{|l|}{$(\ldiap[\psi?] \phi,s)$}        & $\abel$ 
			& $\{ (\psi,s), (\phi,s) \}$ 
			& \multicolumn{2}{|l|}{$(\lboxp[\psi?] \phi,s)$}        & $\eloi$ 
			& $\{ (\mybar{\psi},s), (\phi,s) \}$ 
			\\ \hline
		\end{tabular}
	\end{center}
	\label{tb:EG}
\end{table}
 \vspace{-20pt}

Whenever a position $(\phi,s)$ is reached with $\phi = \bot, \top, p, \mybar{p}$, the game ends and is won by $\eloi$ if $\phi=\top$; $\phi = p$  and $s \in V(p)$; or $\mybar{p}$ and $s\notin V(p)$, and is won by $\abel$ else.
Note that the left projection $(\phi_{n})_{n<\kappa}$ of any (partial) match $(\phi_{n},s_{n})_{n<\kappa}$ is a trace.
An infinite match is won by $\abel$ if its left projection features infinitely many diamond fixpoint formulas (or, equivalently, cofinitely many diamond formulas) and by $\eloi$ else.
(This corresponds to the standard definition of winning conditions for evaluation games for fixpoint formulas, cf.~Proposition~\ref{p:inftr}.)
It is clear that this game can be presented as a parity game, and as such it is positionally determined.

\begin{ourdefinition}
\label{d:gsem}
Let $\bbS, s$ be a pointed model, let $f$ be a strategy for $\eloi$ in $\calE(\bbS)$ and let $\phi$ be a formula.
We write $\bbS, s \sat_f \phi$ if $f$ is winning for $\eloi$ at $(\phi,s)$, and $\bbS,s \sat \phi$ if $\bbS, s \sat_f \phi$ for some strategy $f$ for $\eloi$. 
We say that a formula $\phi$ is \emph{satisfiable} if there exists $\bbS,s$ such that $\bbS,s \sat \phi$ and \emph{unsatisfiable} otherwise. 
\end{ourdefinition}


\section{Proof system}\label{sec.proofSystem}

An \emph{annotated formula} is a pair $(\phi,o)$, usually denoted as $\phi^o$, where $o$ is either $u$ (\emph{unfocused}) or $f$ (\emph{in focus}). An \emph{annotated sequent} $\Gamma$ is a set of annotated formulas, such that at most one formula in $\Gamma$ is in focus. 
Given a set of formulas $\Delta$, we define the annotated sequent $\Delta^u \isdef \{\phi^u \| \phi \in \Delta\}$. 
For an annotated sequent $\Gamma$ we define $\Gamma^- \isdef \{\phi \| \phi^o \in \Gamma\}$ and $\Gamma^u \isdef \{\phi^u \| \phi^o \in \Gamma\}$. We read annotated sequents \emph{conjunctively} and say that $\Gamma$ is \emph{satisfiable} if $\Land \Gamma^-$ is satisfiable and call $\Gamma$ \emph{unsatisfiable} else. 
If no confusion is likely, we will call annotated sequents just sequents.

The rules of the proof system \CPDLf are given in Figure \ref{fig.rulesCPDL}. 
Note that the calculus aims to derive sequents that are \emph{unsatisfiable}.
Apart from the annotations the rules are as expected. In the rules \RuDia, \RuAnd, \RuOr, \RuConD, \RuConB, \RuChoiceD, \RuChoiceB, \RuTestD, \RuTestB, \RuStarD, \RuStarB and \RuWeak we call the single explicitly written formula in its conclusion the \emph{principal formula} of the rule. 
We only allow applications of rules $\Ru$ with a principal formula different from $\ldiap[\alpha]\phi$ if the principal formula is unfocused.
(Hence, to apply such a rule to a formula in focus, first a \RuU rule has to be applied.) The modal rule \RuDiaP is only allowed if its principal formula $\ldiap[a]\phi$ is in focus.

In applications of \RuACut we demand that $\phi \in \FLN(\Gamma)$. For a modal rule \RuDiaP with conclusion $\Theta = \ldiap[a] \phi^{\anColor{f}}, \lboxp[a] \Sigma, \Gamma$ the sequent $\ldiap[\conv{a}]\Gamma$ is defined as $\ldiap[\conv{a}]\Gamma \isdef \{\ldiap[\conv{a}]\chi^u \| \chi^u \in \Gamma \text{ and }\ldiap[\conv{a}]\chi \in \FLN(\Theta)\}$. 
 This ensures that all rules are analytic.
\begin{figure}[htb]
    \begin{mdframed}[align=center]
	\begin{center}{}
		\begin{minipage}{0.29\textwidth}
			\begin{prooftree}
				\hypo{}
				\infer[left label=\AxLit:]1{\phi^{\anColor{o_1}}, \mybar{\phi}^{\anColor{o_2}},\Gamma}
			\end{prooftree}
		\end{minipage}
		\begin{minipage}{0.23\textwidth}
			\begin{prooftree}
				\hypo{}
				\infer[left label=\AxBot:]1{\bot^{\anColor{u}},\Gamma}
			\end{prooftree}
		\end{minipage}
		\begin{minipage}{0.22\textwidth}
			\begin{prooftree}
				\hypo{\phi^{\anColor{u}},\psi^{\anColor{u}}, \Gamma}
				\infer[left label= \RuAnd:]1{\phi \land \psi^{\anColor{u}}, \Gamma}
			\end{prooftree}
		\end{minipage}
		\begin{minipage}{0.22\textwidth}
			\begin{prooftree}
				\hypo{\phi^u, \Gamma}
				\hypo{\psi^u, \Gamma}
				\infer[left label= \RuOr:]2{\phi \lor \psi^u, \Gamma }
			\end{prooftree}
		\end{minipage}
	\end{center}
	
	\begin{center}{}
		\begin{minipage}{0.35\textwidth}
			\begin{prooftree}
				\hypo{\ldiap[\alpha]\ldiap[\beta]\phi^{\anColor{o}}, \Gamma}
				\infer[left label=\RuConD:]1{\ldiap[\alpha;\beta]\phi^{\anColor{o}},\Gamma}
			\end{prooftree}
		\end{minipage}		
		\begin{minipage}{0.29\textwidth}
			\begin{prooftree}
				\hypo{\lboxp[\alpha]\lboxp[\beta]\phi^{\anColor{u}}, \Gamma}
				\infer[left label=\RuConB:]1{\lboxp[\alpha;\beta]\phi^{\anColor{u}},\Gamma}
			\end{prooftree}
		\end{minipage}
		\begin{minipage}{0.32\textwidth}
			\begin{prooftree}
				\hypo{\ldiap[\alpha]\phi^{\anColor{o}}, \Gamma}
				\hypo{\ldiap[\beta]\phi^{\anColor{o}},\Gamma}
				\infer[left label=\RuChoiceD:]2[]{\ldiap[\alpha \cup \beta]\phi^{\anColor{o}},\Gamma}
			\end{prooftree}
		\end{minipage}		
	
	\end{center}
	
	\begin{center}{}
		\begin{minipage}{0.35\textwidth}
			\begin{prooftree}
				\hypo{\ldiap[\alpha] \ldiap[\alpha^*] \phi^{\anColor{o}},\Gamma}
				\hypo{\phi^{\anColor{o}},\Gamma}
				\infer[left label= \RuStarD:]2{ \ldiap[\alpha^*] \phi^{\anColor{o}},\Gamma}
			\end{prooftree}
		\end{minipage}
		\begin{minipage}{0.29\textwidth}
			\begin{prooftree}
				\hypo{\lboxp[\alpha]\lboxp[\alpha^*] \phi^u, \phi^u, \Gamma}
				\infer[left label= \RuStarB:]1{\lboxp[\alpha^*] \phi^u, \Gamma}
			\end{prooftree}
		\end{minipage}
		\begin{minipage}{0.32\textwidth}
			\begin{prooftree}
				\hypo{\lboxp[\alpha]\phi^{\anColor{u}},\lboxp[\beta]\phi^{\anColor{u}},\Gamma}
				\infer[left label=\RuChoiceB:]1{\lboxp[\alpha \cup \beta]\phi^{\anColor{u}},\Gamma}
			\end{prooftree}
		\end{minipage}
		
	\end{center}
		
	\begin{center}{}
		\begin{minipage}{0.35\textwidth}
			\begin{prooftree}
				\hypo{\psi^{\anColor{u}}, \phi^{\anColor{o}}, \Gamma}
				\infer[left label=\RuTestD:]1{\ldiap[\psi?]\phi^{\anColor{o}},\Gamma}
			\end{prooftree}
		\end{minipage}		
		\begin{minipage}{0.29\textwidth}
			\begin{prooftree}
				\hypo{\mybar{\psi}^u, \Gamma}
				\hypo{\phi^u,\Gamma}
				\infer[left label=\RuTestB:]2{[\psi?]\phi^u, \Gamma}
			\end{prooftree}
		\end{minipage}
		\begin{minipage}{0.32\textwidth}
			\begin{prooftree}
				\hypo{\phi^{\anColor{f}}, \Sigma, \ldiap[\conv{a}] \Gamma}
				\infer[left label= \RuDiaP:]1{ \ldiap[a] \phi^{\anColor{f}}, \lboxp[a] \Sigma,\Gamma}
			\end{prooftree}
		\end{minipage}	
	\end{center}

	\begin{center}{}
		\begin{minipage}{0.29\textwidth}
			\begin{prooftree}
				\hypo{\Gamma, \phi^u}
				\hypo{\mybar{\phi}^{\anColor{u}},\Gamma}
				\infer[left label= \RuACut:]2{\Gamma}
			\end{prooftree}
		\end{minipage}
		\begin{minipage}{0.23\textwidth}
			\begin{prooftree}
				\hypo{\Gamma}
				\infer[left label=\RuWeak:]1{\phi^{\anColor{u}},\Gamma}
			\end{prooftree}
		\end{minipage}	
		\begin{minipage}{0.22\textwidth}
		\begin{prooftree}
			\hypo{\phi^{\anColor{f}},\Gamma}
			\infer[left label= \RuF:]1{ \phi^{\anColor{u}}, \Gamma}
		\end{prooftree}
		\end{minipage}
		\begin{minipage}{0.22\textwidth}
			\begin{prooftree}
				\hypo{\phi^{\anColor{u}},\Gamma}
				\infer[left label= \RuU:]1{ \phi^{\anColor{f}}, \Gamma}
			\end{prooftree}
		\end{minipage}			
	\end{center}
    \end{mdframed}
	\caption{Rules of \CPDLf}
	\label{fig.rulesCPDL}
\end{figure}
 \vspace{-10pt}

\begin{ourdefinition}[Derivation]\label{def.CPDLfDerivation}
A \emph{\CPDLf derivation} $\pi = (T,\edge,\Plab,\sfR)$ is a proof tree defined from the rules in Figure \ref{fig.rulesCPDL} such that $(T,\edge)$ is a, possibly infinite, tree with nodes $T$ and parent relation\footnote{We write $s \edge t$ if $s$ is the parent of $t$.} $\edge$;
$\Plab$ maps each node $u \in T$ to an annotated sequent $\Plab_u$; and	$\sfR$ is a function that maps every node $u \in T$ to either (i) the name of a rule in Figure \ref{fig.rulesCPDL} or (ii) an extra value $e$, such that every node labelled with $e$ is a leaf.	Leaves labelled by $e$ are called \emph{assumptions}.
We require the maps $\Plab$ and $\sfR$ to be in accordance with the formulation of the rules, as usual.
	
An assumption $l$ is called a \emph{repeat leaf}  if it has a proper ancestor $v$ of $l$ that such that $\Plab_l = \Plab_v$. In this case the nearest such ancestor of $l$ is called its \emph{companion} and we denote it by $c(l)$.
\end{ourdefinition}
	
Given a \CPDLf derivation $\pi = (T,\edge,\Plab,\sfR)$ we define the usual \emph{proof tree} $\calT_\pi = (T,\edge)$ and the \emph{proof tree with back edges} $\cyclicPT = (T,\cEdge)$, where $\cEdge = \edge \cup \{(l,c(l)) \mid l \text{ is a repeat leaf}\}$. 
Note that we think of trees as growing \emph{upwards}. We call $\pi$ \emph{regular} if it has finitely many distinct subderivations. 
We define $\edgeT$ and $\edgeRT$ to be the transitive and reflexive-transitive closures of $\edge$, respectively, and $\cEdgeT$ and $\cEdgeRT$ to be the transitive and reflexive-transitive closures of $\cEdge$, respectively.

\begin{ourdefinition}
A finite path $\tau$ in a \CPDLf derivation is \emph{successful} if 
\begin{enumerate}
\item every sequent on $\tau$ has a formula in focus, and 
\item there is a node on $\tau$ where the formula in focus is principal.
\end{enumerate}
	
Let $l$ be a  repeat leaf in a \CPDLf derivation  $\pi$ with companion  $c(l)$, and let $\tau_l$ denote the \emph{repeat path of $l$} in $\calT_{\pi}$ from $c(l)$ to $l$. 
We say that $l$ is a \emph{discharged assumption} if the path $\tau_l$ is successful.		
A leaf is called \emph{closed} if it is either a discharged assumption or labelled by an axiom, and is called \emph{open} otherwise.
\end{ourdefinition}

\begin{ourdefinition}[Cyclic proof]
\label{def.CPDLfProof}
Let $\calA$ be a set of sequents. A \CPDLf \emph{proof with assumptions} $\calA$ is a finite \CPDLf derivation $\pi$, where every leaf of $\pi$ is either closed or labeled by a sequent in $\calA$.
If the root of such a proof $\pi$ is labeled with a sequent $\Gamma$, we write $\pi:\calA \proves \Gamma$.
We say that \CPDLf \emph{proves} a sequent $\Gamma$ with assumptions $\calA$ and write $\calA \proves \Gamma$, if we have $\pi:\calA \proves \Gamma$ for some $\pi$.
If $\calA$ is empty, we write $\proves \Gamma$. 
For a set of unannotated formulas $\Delta$ we define $\proves \Delta$ as $\proves \Delta^u$. 
\end{ourdefinition}

The following theorem states the soundness and completeness of this system.
This result follows as a special case of the soundness and completeness of the split proof system that we introduce in the next section.
\begin{theorem}[Soundness \& Completeness]
A sequent $\Gamma$ is unsatisfiable iff $\proves \Gamma$.
\end{theorem}


\section{Split proof system}

One of the core ideas underlying Maehara's proof-theoretic approach towards Craig interpolation is to work with a version of the derivation system that operates on so-called \emph{split sequents}.
Here a \emph{split sequent} $(\Gamma, \Delta)$, usually written as $\Gamma \| \Delta$, is a pair of annotated sequents, of which at most one can be focused.
Note that we do not require $\Gamma$ and $\Delta$ to be disjoint.
Given a split sequent $\Sigma = \Gamma \| \Delta$, we will write $\Sigma^l$ for its \emph{left component} $\Gamma$ and $\Sigma^r$ for its \emph{right component} $\Delta$. We will use $d$ as a variable ranging over the set $\{l,r\}$. 
If $\Gamma$ (respectively $\Delta$) contains a formula in focus, we call $\Gamma$ (respectively $\Delta$) 
the \emph{focused component} of $\Gamma \| \Delta$. 

The rules of the split proof system $\SCPDLf$ are obtained from the rules of \CPDLf by applying the rules to one of the components. 
Formally, let \newline $\begin{prooftree}
	\hypo{\Gamma_1,\Delta_1}
	\hypo{\cdots}
	\hypo{\Gamma_n,\Delta_n}
	\infer[left label= \Ru]3[]{\Gamma, \Delta}
\end{prooftree}$ be a rule of Figure \ref{fig.rulesCPDL}. 
Then $\begin{prooftree}
	\hypo{\Gamma_1\|\Delta_1}
	\hypo{\cdots}
	\hypo{\Gamma_n \| \Delta_n}
	\infer[left label= $\Ru^l$]3[]{\Gamma\| \Delta}
\end{prooftree}$ is a \emph{left rule}, if the following conditions are satisfied:
\begin{enumerate}
	\item if \Ru is not an axiom and $\Ru \neq \RuDia$, then $\Delta_i = \Delta$ for $1 \leq i \leq n$,
	\item if \Ru = \RuACut, then $\phi \in \FLN(\Gamma)$,
	\item if \Ru = \RuDia, then $\RuDia^l$ is of the form 
	$\begin{prooftree}
		\hypo{\phi^{\anColor{f}}, \Sigma, \ldiap[\conv{a}] \Lambda \mid \Pi, \ldiap[\conv{a}]\Theta}
		\infer[left label= $\RuDiaP^{l}:$]1{ \ldiap[a] \phi^{\anColor{f}}, \lboxp[a] \Sigma, \Lambda \mid \lboxp[a] \Pi, \Theta}
	\end{prooftree}$
	where $\ldiap[\conv{a}]\Lambda \isdef \{\ldiap[\conv{a}]\chi^u \mid \chi^u \in \Lambda \text{ and } \ldiap[\conv{a}]\chi \in \FLN(\Gamma)\}$ and 
	$\ldiap[\conv{a}]\Theta \isdef \{\ldiap[\conv{a}]\chi^u \mid \chi^u \in \Theta \text{ and } \ldiap[\conv{a}]\chi \in \FLN(\Delta)\}$.\footnote{Note that here $\Gamma = \ldiap[a] \phi^{\anColor{f}}, \lboxp[a] \Sigma, \Lambda$ and $\Delta = \lboxp[a] \Pi, \Theta$.}
\end{enumerate}
\emph{Right rules} are defined analogously. \emph{Split rules} are either left or right rules.  Note that the only split rules with interactions between the components are axioms and modal rules.
Notions defined in \Cref{sec.proofSystem} translate straightforwardly to \SCPDLf derivations.

\begin{ourdefinition}\label{def.SCPDLfDerivation}
An \emph{\SCPDLf derivation} is defined analogously to a \CPDLf derivation, where each node is now labelled by a split sequent and a split rule. 
A finite path $\tau$ in an \SCPDLf derivation is \emph{successful} if every node on $\tau$ features a formula in focus and there is a node on $\tau$ where the formula in focus is principal.
An \emph{\SCPDLf proof with assumptions $\calA$} is a finite \SCPDLf derivation, where every leaf is either closed or labelled by a sequent in $\calA$. 
The relation $\proves$ is defined as for \CPDLf proofs.
\end{ourdefinition}


\section{Soundness and completeness of split proofs}\label{sec.splitCompletness}

We prove the soundness and completeness of \SCPDLf by \emph{game-theoretic} means.
Given a split sequent $\Sigma$, we define a proof-search game $\calG(\Sigma)$ in which one player (Prover) aims to find a proof of $\Sigma$, while the other player (Builder) aims to construct a model where $\Sigma$ is satisfied. 
Winning strategies for Prover and Builders then correspond to, respectively, proofs and models for $\Sigma$.
To tighten the correspondence between winning strategies for Prover and proofs we will work with infinite \SCPDLInfty proofs. These are \SCPDLf derivations, where instead of cycles we allow infinite branches satisfying a corresponding success condition.

\subsection{Infinite $\SCPDLInfty$ proofs}

\begin{ourdefinition}
	An infinite path $\tau$ in an \SCPDLf derivation is \emph{successful} if
	\begin{enumerate}
		\item on cofinitely many sequents on $\tau$ there is a formula in focus, and
		\item there are infinitely many nodes on $\tau$ where the formula in focus is principal.
	\end{enumerate}
\end{ourdefinition}

\begin{ourdefinition}[Infinitary proof]
	An \SCPDLInfty proof is an \SCPDLf derivation, where every leaf is labelled by an axiom and every infinite path is successful.
\end{ourdefinition}

The correspondence between regular \SCPDLInfty proofs and \SCPDLf proofs is standard in cyclic proof theory.
\begin{lemma}\label{lem.CPDLfIffCPDLInfty}
There is a regular \SCPDLInfty proof of $\Gamma\| \Delta$ iff $\SCPDLf \proves \Gamma \| \Delta$.
\end{lemma}
\begin{proof}
	Let $\rho$ be a regular $\SCPDLInfty$ proof. 
	For a node $v$ in $\rho$ let $\rho_v$ be the subderivation of $\rho$ rooted at $v$. 
	For every infinite path $\tau = (\tau(i))_{i \in \omega}$ define minimal 
	indices $j < k$ such that
	\begin{enumerate}
		\item $\rho_{\tau(j)} = \rho_{\tau(k)}$ and
		\item the path $\tau(j)\cdots \tau(k)$ is successful.
	\end{enumerate}
	Because $\rho$ is regular and every infinite path is successful such indices always exist. For each such infinite path let  $\tau(k)$ be a leaf discharged with companion $\tau(j)$. Using König's Lemma we can show that this procedure results in a finite \SCPDLf proof $\pi$ of $\Gamma$.
\end{proof}

\subsection{Proof search game}

A \emph{rule instance} is a triple $(\Sigma,\Ru, \langle\Sigma_1,...,\Sigma_n\rangle)$ such that
$\begin{prooftree}
	\hypo{\Sigma_1}
	\hypo{\cdots}
	\hypo{\Sigma_n}
	\infer[left label = \Ru]3[]{\Sigma}
\end{prooftree}$
is a valid rule application in \SCPDLf. We let $\conc$ be the function mapping rule instances to their conclusions.
A rule instance is \emph{cumulative} if all premises are componentwise supersets of the conclusion and \emph{productive} if each premise is distinct from the conclusion.
We call a rule instance of \RuU \emph{conceding}, if its principal formula is of the form $\ldiap[\alpha]\phi$ and \emph{unblocking} otherwise.

Let $\Phi$ be a split sequent. We define the proof search game $\calG(\Phi)$, with players Prover and Builder. 
Its positions are given by $\Seq_{\Phi}\cup \Inst_{\Phi}$, where $\Seq_{\Phi}$ is the set of split sequents and $\Inst_{\Phi}$ the set of rule instances containing formulas in $\FLN(\Phi)$. The ownership function and admissible moves are given in the table below. An infinite match is won by Prover iff the resulting infinite path is successful.
\vspace{-15pt}
\begin{table}[H]
	\begin{center}
		\begin{tabular}{|c|c|c|}
			\hline
			Position & Owner & Admissible moves \\
			\hline
			$\Sigma$ & Prover & $\{i \in \Inst_{\Phi} \| \conc(i) = \Sigma\}$ \\
			$(\Sigma,\sfR, \langle\Sigma_1,\cdots,\Sigma_n\rangle)$ & ~Builder~ & $\{\Sigma_i \| i = 1,\cdots,n\}$ \\
			\hline
		\end{tabular}
	\end{center}
\end{table}
\vspace{-25pt}


\noindent
Interestingly, we may assume that winning strategies in $\calG(\Phi)@\Phi$ are positional, that is, only depend on the current position of the game, and not on the history of the play leading up to this position.
The key observations here is that $\calG(\Phi)$ can be formulated as a parity game, since parity games are well known to enjoy positional determinacy~\cite{Emerson99,Mostowski1991}.
To see that $\calG(\Phi)$ is a parity game, we may assign the following priorities to its positions: $\Omega(\Sigma) \isdef 0$ for any sequent position, and for the other positions we put
\[\Omega(\Sigma,\sfR, \langle\Sigma_1,\cdots,\Sigma_n\rangle) \isdef 
   \left\{\begin{array}{ll}
      3 & \text{ if $\Sigma$ has no formula in focus}
   \\ 2 & \text{ if the principal formula of $\sfR$ is in focus}
   \\ 1 & \text{ otherwise}.
   \end{array}\right.
\]
An \SCPDLInfty proof of a split sequent $\Phi$ may be identified with a winning strategy of Prover in $\calG(\Phi)@\Phi$, and thus, as a consequence of positional determinacy we may assume that \CPDLInfty proofs are regular.

\subsection{Soundness}

Given a \CPDLInfty proof of a split sequent $\Sigma = \Gamma \| \Delta$ we want to show that $\Gamma \cup \Delta$ is unsatisfiable. Equivalently, given a pointed model $\bbS,s$ that satisfies $\Gamma\cup \Delta$ we provide a winning strategy for Builder in $\calG(\Sigma)@(\Sigma)$. The proof is standard and can be found in Appendix \ref{app.splitCompleteness}.

\begin{theorem}[Soundness]\label{thm.soundnessSplit}
	If $\SCPDLf \proves \Gamma \| \Delta$, then $\Gamma, \Delta$ is unsatisfiable.
\end{theorem}

\subsection{Completeness}

For any unsatisfiable sequent $\Gamma, \Delta$ we have to find a \SCPDLInfty proof of $\Sigma = \Gamma \| \Delta$, in other words a winning strategy of Prover in $\calG(\Sigma)@\Sigma$. We will show an even stronger statement: By restricting the strategy of Prover we show that for every unsatisfiable split sequent we obtain an \SCPDLInfty proof in a certain normal form. 
These \emph{uniform} proofs will be instrumental in our interpolation proof.

\begin{ourdefinition}[Uniform Split Derivation]\label{def.uniformSplitDerivation}
	A set of formulas $\Gamma$ is called \emph{saturated}, if no axiom or cumulative and productive \CPDLf rule may be applied to $\Gamma^u$. This definition is equivalent to $\Gamma^u$ being a saturated set in the usual sense.
	A split derivation $\pi$ is \emph{uniform} if it satisfies the following conditions:
	\begin{itemize}
		\item[U0.] If possible an axiom is applied.
		\item[U1.] Else if possible a cumulative and productive rule is applied to a formula in an unfocused component. 
		\item[U2.] 
		Let $t_1$ and $t_2$ be nodes labelled with split sequents $\Gamma_1 \| \Delta$ and $\Gamma_2 \| \Delta$ in $\pi$, respectively.
Assume that their common component $\Delta$ is focused, while their unfocused components $\Gamma_1$ and $\Gamma_2$ are both saturated.
Then at both $t_1$ and $t_2$ the same rule with the same principal formula in $\Delta$ is applied. If possible, this rule is cumulative and productive with an unfocused principal formula, and else it is productive with its principal formula in focus.
\item[U3.] The analogous condition to (U2) for split sequents $\Gamma \| \Delta_1$ and $\Gamma \| \Delta_2$.
\end{itemize}
\end{ourdefinition}

We thus aim to show the following completeness theorem. We present a proof sketch; a full proof can be found in Appendix \ref{app.splitCompleteness}.
\begin{theorem}[Completeness]\label{thm.completenessSplit}
	If $\Gamma,\Delta$ is unsatisfiable then there is a uniform \SCPDLf proof of $\Gamma \| \Delta$.
\end{theorem}
\begin{proof}[Sketch]
Let $\Sigma = \Gamma \| \Delta$. By contraposition, given a winning strategy for 
	Builder in $\calG(\Sigma)@\Sigma$ we have to show that $\Gamma,\Delta$ is satisfiable. 
	Let $f$ be a \emph{positional} winning strategy for Builder in $\calG(\Sigma)@\Sigma$, we construct a pointed model $\bbS^f,s$ and a strategy $\underline{f}$ for $\eloi$ in $\calE(\bbS^f)$ such that $\bbS^f,s \sat_{\underline{f}} \Sigma$. 
	
	Let $\calT$ be the subtree of the game-tree of $\calG(\Sigma)@\Sigma$, where Builder plays the strategy $f$ and Prover picks rule instances such that the uniformity conditions are satisfied. We want to define a model $\bbS^f$ from $\calT$.
	A maximal path $\rho$ in $\calT$ not containing rule instances of \RuDia, \RuF and conceding instances of \RuU is called a \emph{local path}. It follows that every local path is finite: Only finitely many cumulative and productive rule instances may be applied, and in the only non-cumulative rule instances that we apply the principal formula is in focus, hence there can also only be finitely many as $f$ is a winning strategy for Builder. 
	
	For local paths $\rho, \tau$ we define $\rho \overset{a}{\to} \tau$ if $\tau$ is above $\rho$ in $\calT$ only separated by an instance of \RuDia[a] and (possibly) instances of \RuF and conceding instances of \RuU. We let $\Plab(\rho) \isdef \bigcup\{\Sigma^l \cup \Sigma^r \| \Sigma \text{ occurs in } \rho\}$.
	
	We now define the model $\bbS^f = (S^f,R^f,V^f)$. We let $S^f$ be the set of local paths in $\calT$, define $V^f(p) \isdef \{\rho \in S^f \| p \in \Plab(\rho)^-\}$ and $R^f = \{R_a^f\}_{a \in \Act}$, where
	\[\rho R_a \tau \quad :\Leftrightarrow \quad \rho \overset{a}{\to} \tau \text{ or } \tau \overset{\conv{a}}{\to} \rho\]
	For the definition of the strategy $\underline{f}$ for $\eloi$ in $\calE(\bbS^f)$ we use the fact that $\Plab(\rho)$ is a saturated set for every local path $\rho$. For instance, at position $(\phi \lor \psi, \rho)$ the formula $\phi$ or $\psi$ is in $\Plab(\rho)^-$ and $\underline{f}$ picks one that is. At position $(\ldiap[a]\phi, \rho)$ the strategy $\underline{f}$ picks some $\tau$ such that $\rho \tom{a} \tau$. 
	
	Now let $\psi_0 \in \Sigma$ and let $\rho_0$ be a local path containing $\Sigma$. Let $\calM$ be an $\underline{f}$-guided $\calE(\bbS^f)$-match starting at $(\psi_0,\rho_0)$. Then we can prove that for every position $(\psi,\rho)$ in $\calM$ it holds that $\psi \in \Plab(\rho)$ and consequently that $\eloi$ wins $\calM$. This shows that $\bbS^f, \rho_0 \sat_f \Sigma$.
\end{proof}

\section{Interpolation}
\label{s:itp}

As mentioned in the introduction, the following theorem is the main contribution of this paper.

\begin{theorem}[Craig Interpolation]
\label{t:itp}
Let $\phi$ and $\psi$ be \CPDL-formulas such that $\phi \models \psi$.
Then there is a \CPDL-formula $\theta$ such that $\Voc(\theta) \subseteq \Voc(\phi) \cap \Voc(\psi)$ and $\phi \models \theta$ and $\theta \models \psi$.
\end{theorem}

As an immediate consequence of this we obtain Beth definability.
Where $\phi(p)$ is a \CPDL-formula, we use $\phi(p_i)$ as an abbreviation of $\phi[p_i/p]$.

\begin{corollary}[Beth Definability]
Let $\phi(p)$ be a \CPDL-formula, and let $p_0,p_1$ be fresh variables.
If $\phi(p_0),\phi(p_1) \models p_0 \liff p_1$, then there is a \CPDL-formula $\chi$ with $\Voc(\chi) \subseteq \Voc(\phi)\setminus\{p\}$ and $\phi(p) \models p \liff \chi$.
\end{corollary}
\begin{proof}
Apply Craig interpolation to $\phi(p_0), p_0 \models \phi(p_1) \impl p_1$.
\end{proof}

The remainder of the paper is devoted to the proof of Theorem~\ref{t:itp}.
In this section we will use the split sequent system to find Craig interpolants for $\CPDL$.
We first transfer the concept of interpolation from formulas to split sequents,
calling a formula $\theta$ an \emph{interpolant} for an unsatisfiable split sequent $\Gamma \| \Xi$ if $\Voc(\theta) \subseteq \Voc(\Gamma) \cap \Voc(\Xi)$, and both sequents $\Gamma \| \theta$ and $\mybar{\theta} \| \Xi$ are unsatisfiable.
Since we have $\phi \models \psi$ iff $\mybar{\psi} \| \phi$ is unsatisfiable, it is easy to see that a formula $\theta$ is an interpolant for the formulas $\phi$ and $\psi$ iff it is an interpolant for the split sequent $\mybar{\psi} \| \phi$.
By the completeness of the split sequent system it therefore suffices to prove the following result.

\begin{theorem}
\label{t:main}
If $\vdash \Gamma \| \Xi$ then the split sequent $\Gamma \| \Xi$ has an interpolant.
\end{theorem}

The proof of Theorem~\ref{t:main} easily follows from the Lemmas~\ref{l:itp1} and~\ref{l.clusterInterpolation} below.
The key notion in the proof is that of a \emph{cluster}; to introduce clusters, let $s, t$ be nodes in a uniform $\SCPDLf$ proof $\pi$.
We define the relation $\equiv_c$ by putting:
\[
s \cEquiv t     \text{ iff }  s \cEdgeRT t \text{ and } t \cEdgeRT s. 
\]
It is easily verified that the relation $\cEquiv$ is an equivalence relation, whose equivalence classes, called \emph{clusters}, are the maximal strongly connected components of $\cEdge$.
Non-singleton clusters are called \emph{proper}.
Note that for any pair $s,t$ of nodes of a proper cluster there are 
$\cEdge$-paths from $s$ to $t$ and vice versa.
Every proper cluster of $\pi$ is in fact a \emph{subtree} of the underlying tree of $\pi$, in the sense that the structure $(C,\edge{\upharpoonright_C})$ is a tree itself (here $\edge{\upharpoonright_C}$ denotes the parent-child relation of $\pi$, restricted to $C$).
In particular, every proper cluster has a root.
We refer to the children of $C$-nodes which lie outside of $C$ as the \emph{exit nodes} of $C$ --- in the case of a singleton cluster these are just the children of the cluster's unique member.

We prove Theorem~\ref{t:main} by induction on the size of the derivation of $\Gamma \| \Xi$.
Lemma \ref{l:itp1} takes care of leaves and of the induction step in the case where the root of the derivation forms a singleton cluster.
We omit its proof, which is a straightforward adaptation of Maehara's method for well-founded proofs.

\begin{lemma}\label{l:itp1}
Let $\pi$ be an $\SCPDLf$ proof of $\Gamma \| \Xi$.
Assume that the root $r$ of $\pi$ forms a singleton cluster, and that for every child $t$ of $r$ we have an interpolant $\theta_{t}$ for the 
split sequent $\Plab^{l}_{t} \| \Plab^{r}_{t}$.
Then $\Gamma \| \Xi$ has an interpolant.
\end{lemma}

The key task is to obtain interpolants for the roots of proper clusters.

\begin{lemma}\label{l.clusterInterpolation}
Let $\pi$ be a uniform $\SCPDLf$ proof of $\Gamma \| \Xi$, assume that the root $r$ of $\pi$ belongs to a proper cluster $C$, and that for every exit node $t$ of $C$ we have an interpolant $\theta_{t}$ for the split sequent $\Plab^{l}_{t} \| \Plab^{r}_{t}$.
Then $\Gamma \| \Xi$ has an interpolant.
\end{lemma}

\begin{center}
Fix $\pi, C, r, \Gamma, \Xi$ for the remainder of the paper.
\end{center}
For reasons of symmetry we may confine our attention to the case where $\Xi$ is focused.
Furthermore, we will assume that $\Gamma$ is non-empty, since if $\Gamma = \emptyset$ we may simply define the interpolant to be the formula $\bot$.

\subsubsection{Proper clusters}
We first discuss proper clusters in some more detail.
Let
\[
C^{+} \isdef C \cup \{ s \in \pi \mid t \edge s, \text{ for some } t \in C\}
\]
be the set of nodes that either belong to $C$ or are the child of a node in $C$.
Then $C^+ \setminus C$ is the set of exit nodes of $C$.
The following lemma will be used implicitly.

\begin{lemma}\label{l:itp4}
For all $t \in C$, the following hold: (1) $\Plab^{r}_{t}$ and $\Plab^{l}_{t}$ are both non-empty; (2) $\Plab^{r}_{t}$ is focused; (3) all children of $t$ are in $C^{+}$, and at least one is in $C$; (4) if a right rule is applied at $t$ then $\Plab^{r}_{t} \neq \Plab^{r}_{u}$, for every child $u$ of $t$.
\end{lemma}

\begin{ourdefinition}
We let $\FRC$ denote the sets of sequents occurring as a right component in $C$, namely $\FRC \isdef \{ \Plab^{r}_{t} \mid t \in C\}$, and likewise for $\FRCp$.
Given a sequent $\Delta \in \FRCp$, we define $
   C_{\Delta} \isdef \{ t \in C \mid \Plab^{r}_{t} = \Delta \}$, 
 $C_{\Delta}^{+} \isdef \{ t \in C^{+} \mid \Plab^{r}_{t} = \Delta \}$
and we let $C_{\Delta}^{l}$ ($C_{\Delta}^{r}$, respectively), denote the set of nodes in $C_{\Delta}$ where a left rule (a right rule, respectively) is applied.
\end{ourdefinition}


\begin{lemma}
For all sequents $\Delta$, the following hold:
\begin{enumerate}
\item $C_{\Delta} = C_{\Delta}^{l} \uplus C_{\Delta}^{r}$.
\item If $t \in C_{\Delta}^{l}$ then the rule applied at $t$ is not the modal rule.
\item If $t \in C_{\Delta}^{l}$ then all of its children belong to $C_{\Delta}^+$, and at least one to $C_{\Delta}$.
\item If $C_\Delta$ is not empty then $C^{r}_\Delta$ is not empty. 
\end{enumerate}
\end{lemma}

By uniformity of $\pi$, for each $\Delta \in \FRC$ there is a unique right rule $R_{\Delta}$ which is applied at each $t \in C_{\Delta}^{r}$ (provided $C_{\Delta}^{r} \neq \emptyset$). 
If $R_{\Delta}$ is the modal rule we call $\Delta$ a modal sequent.
If $\Delta \in \FRCp \setminus \FRC$, we call $\Delta$ an \emph{exit sequent}.

\begin{lemma}\label{l.uniformChildren}
If $\Delta$ is neither a modal nor an exit sequent then there are sequents $\Pi_{1},\ldots,\Pi_{n}$ such that $\bigwedge \Delta \equiv \bigvee_{i} \bigwedge \Pi_{i}$, and, for all $t \in C_{\Delta}^{r}$, the children of $t$ can be listed as $u_{1},\ldots, u_{n}$ such that $\Plab_{u_{i}} = \Plab^{l}_{t}\|\Pi_{i}$, for all $i=1,...,n$.
\end{lemma}

\subsubsection{Quasi-proofs}
We can now introduce the pivotal structure in our interpolation proof: the \emph{quasi-proof} $\sfQ = (Q, \qedge, k, \Qlab)$ associated with the cluster $C$.
Roughly, $\sfQ$ is a finite labeled tree that represents the focused part of $C$.
In particular, its labeling is a map $\Qlab: Q \to \FRCp$ that respects the labeling of $C$ as suggested by Lemma~\ref{l.uniformChildren}; also, any node labeled with an exit sequent is a leaf of $\sfQ$.
To ensure that $\sfQ$ is based on a \emph{finite} tree, we make sure that every repeat node is a leaf.

To explain the role of the typing map $k$ in $\sfQ$, note that the purpose of $\sfQ$ is to help find an interpolant for the root $r$ of $C$.
We will do this by inductively associating with each node in $\sfQ$ an auxiliary formula that we will call a \emph{pre-interpolant}.
To facilitate this definition we construct $\sfQ$ in such a way that its internal nodes come in \emph{triples}.
The subsequent nodes of such a triple are all labeled with the same sequent in $\FRCp$, but they have a different \emph{type} (respectively, 1, 2 and 3).
This typing will play a role in the actual definition of the pre-interpolants.

\begin{ourdefinition}
Given the cluster $C$ we construct a structure $\sfQ = (Q, \qedge, k, \Qlab)$, called a \emph{quasi-proof}, step by step.
Here $(Q, \qedge)$ will be a finite tree, $k: Q \to \{ 1, 2, 3 \}$ types the nodes of  $\sfQ$ and $\Qlab: Q \to \FRCp$ is a labeling.

To start the construction, we put a root node $r_{\sfQ}$ in $\sfQ$ and let
$r_{\sfQ}$ have type 1 and label $\Xi$.
Inductively, given a node $x \in Q$, define the children of $x$ as follows:
\begin{description}
\item[Case $k(x) = 1$.]
If $x$ is a \emph{repeat} in $\sfQ$ (i.e., there exists $y \in Q$ such that $y$ is an ancestor of $x$ and $\Qlab_x=\Qlab_y$) or an exit (i.e., $\Qlab_{x}$ is an \emph{exit} sequent), then $x$ is a leaf.
Otherwise, $x$ has a unique child with type 2 and label $\Qlab_{x}$.

\item[Case $k(x) = 2$.]
Then $x$ has a unique child with type 3 and label $\Qlab_{x}$.

\item[Case $k(x) = 3$.] 
In this case\footnote{%
    Every type-3 node is the grandchild of a type-1 node with the same label,
	but any type-1 node $z$ such that $C^{r}_{\Qlab_{z}} = \emptyset$ is a leaf.
	}	
$C^{r}_{\Qlab_{x}} \neq \emptyset$.
If $\Qlab_{x}$ is modal, say, it is of the form $\ldiap[a]\phi^f, 
\lboxp[a]\Sigma, \Pi$ then $x$ has a unique child $y$ with type 1 and label 
$\phi^f, \Sigma, \ldiap[\conv{a}]\Pi$.

If $\Qlab_{x}$ is not modal then by \Cref{l.uniformChildren} there exist 
$\Pi_{1}, \dots, \Pi_{n}$ such that $\bigwedge \Qlab_x \equiv \bigvee_{i} 
\bigwedge \Pi_{i}$ and for all $t \in C_{\Qlab_x}^{r}$, the children of
$t$ can be listed as $u_{1}, \ldots, u_{n}$ with $\Plab_{u_{i}} = 
\Plab^{l}_{t}\|\Pi_{i}$, for all $i$.
We define the children of $x$ in $\sfQ$ as $y_{1}, \ldots, y_{n}$, where each
$y_{i}$ has type 1 and label $\Pi_{i}$.
\end{description}
Given the repeat condition it is fairly easy to check that $\sfQ$ is a finite tree.
\end{ourdefinition}

\begin{ourdefinition}
We let $<_{\sfQ}$ and $\leq_{\sfQ}$ denote, respectively, the
transitive and the reflexive-transitive closure of $\qedge$.
For a repeat leaf $z \in Q$, we let $c(z)$ be the \emph{companion} of $z$, defined as the unique node $x$ such that $x <_{\sfQ} z$, $k(x) = k(z) = 1$ and $\Qlab_{x} = \Qlab_{z}$.
We let $L_{\sfQ}$ and $K_{\sfQ}$ denote, respectively, the sets of repeats and companions of $\sfQ$, in particular $K_{\sfQ} \isdef \{ c(z) \mid z \in L_{\sfQ}\}$.
Given $x \in Q$ we set
\[
  \cycs{x} \isdef \{ z \in L_{\sfQ} \mid 
  c(z) <_{\sfQ} x \leq_{\sfQ} z \}
~\text{ and }~
  \K{x} \isdef \{ c(z) \mid z \in \cycs{x} \}.
\]
A \emph{repeat path} in $\sfQ$ is a sequence of the form 
$(x_{k})_{0\leq k \leq n}$ such that for some leaf $z$ we have $x_{0} = c(z)$,
$x_{n} = z$ and $x_{i} \qedge x_{i+1}$ for all $i<n$.
\end{ourdefinition}

\begin{lemma}\label{l.KandLProps}
Let $\sfQ$ be a quasi-proof and let $x \in Q$. Then the following hold:
\begin{enumerate}
\item $\Qlab_x$ is focused.
\item If $x$ is a leaf or a companion node in $\sfQ$, then $k(x) = 1$.
\item $\cycs{r_{\sfQ}} = \K{r_{\sfQ}} = \emptyset$.
\item $x \notin \K{x}$.
\item
If $x$ is not a companion then $\cycs{x} = \bigcup_{x \qedge y} \cycs{y}$ and $\K{x} = \bigcup_{x \qedge y} \K{y}$.
\item 
If $x$ only has one child, say, $y$, then $\cycs{x} \subseteq \cycs{y}$ and $\K{x} \subseteq \K{y}$.
\item \label{eqit:6}
Every repeat path features nodes with distinct formulas in focus.
\end{enumerate}
\end{lemma}
\begin{proof}
The items 1 -- 6 follow immediately from the definitions.
For item~\ref{eqit:6} let $\rho$ be a repeat path. We write $\rho = x_{0}y_{0}z_{0}x_{1}y_{1}\cdots z_{n-1}x_{n}$, where the $x, y$ and $z$-nodes are, respectively, of type $1, 2$ and $3$.
Note that $\Qlab_{x_{i}} = \Qlab_{y_{i}} = \Qlab_{z_{i}}$ for all $i<n$, and that $\Qlab_{x_{0}} = \Qlab_{x_{n}}$.  
We will simply write $\Qlab_{i}$ for $\Qlab_{x_{i}}$ and let $\xi_{i}$ denote the formula in focus in $\Qlab_{i}$.
The key claim in the proof is the following:
\begin{equation}
\label{eq:pr1}
\text{if $\xi_{i} = \xi_{i+1}$ then $\Qlab_{i} \subset \Qlab_{i+1}$.}
\end{equation}
To see this, first note that by definition of $\sfQ$ there must be some $t \in C^{r}_{\Qlab_{i}}$ which has a successor $u \in C^{r}_{\Qlab_{i+1}}$.
Hence by Lemma~\ref{l:itp4}(4) the sets $\Qlab_{i} = \Plab^{r}_{t}$ and $\Qlab_{i+1} = \Plab^{r}_{u}$ must be distinct.
Now assume $\xi_{i} = \xi_{i+1}$; it follows that the principal formula at $t$ must be unfocused.
But by uniformity the rule applied to this formula is cumulative and productive, which implies that $\Qlab_{i}$ is a proper subset of $\Qlab_{i+1}$.
This proves \eqref{eq:pr1}.

Now assume for contradiction that $\xi_{i} = \xi_{i+1}$ for all $i<n$.
Then $(\Qlab_{i})_{0\leq i \leq n}$ is a strictly increasing sequence of sets, which clearly contradicts the assumption that $\Qlab_{0} = \Qlab_{n}$.

\end{proof}


As mentioned before, each node $x$ in $\sfQ$ represents a certain (not necessarily connected) subset $R_x$ of $C^+$, which we call its \emph{region}:
\[
R_{x} \isdef \left\{\begin{array}{ll}
   C_{\Qlab_x}^{+} & \text{if } k(x) = 1,2
\\[1mm] C_{\Qlab_x}^{r} & \text{if } k(x) = 3.
\end{array}\right.
\]

\begin{lemma}\label{l.uniformChildrenSimulation}
Let $x \in \sfQ$ with $k(x) = 3$.
If we list the children of $x \in \sfQ$ as $z_1, \dots, z_n$ then for all
$t \in R_x$, the children of $t$ may be listed as $s_1, \dots, s_n$ so that 
$s_i \in R_{z_i}$ for all $1 \leq i \leq n$.
\end{lemma}

\subsubsection{Pre-interpolants and the interpolant}
We are now ready to define the interpolant $\theta_{r}$ for the root $r$ of $C$.
The key idea underlying this definition is to first associate with each node $x$ in the quasi-proof $\sfQ$ a so-called \emph{pre-interpolant} $\iota_{x}$.
These pre-interpolants are auxiliary formulas that will be defined by a leaf-to-root induction on the tree $(Q,\qedge)$; once we have arrived at the root $r_{\sfQ}$ of $\sfQ$ we simply define the interpolant $\theta_{r}$ as $\theta_r \isdef \iota_{r_{\sfQ}}$.
For the definition of these pre-interpolants we extend the language with a set
$
\{ q_{x} \mid x \in K_{\sfQ} \}
$
of internal variables, and with every set $\Delta \in \FRCp$  we associate an `exit interpolant':
\[
\theta_{\Delta} \isdef \bigwedge \{ \theta_{t} \mid t \in C^{+}_{\Delta}\setminus C_{\Delta}\}.
\]


\begin{ourdefinition}\label{d.interpolant}
By a leaf-to-root induction, we define for all nodes $x, y \in Q$ a formula $\psi_{x}$ and a family of programs $\{\alpha_{x,y} \mid y \in \K{x}\}$. In all applicable cases, $z$ denotes the unique successor of $x$, and in the case where $x$ is a modal node of type 3, $a$ denotes the leading atomic program of the formula in focus in $\Qlab_{x}$.  Note that for exit nodes $x$, we have $\K{x} = \emptyset$, so no definition of $\alpha_{x, y}$ is required.
\[
\begin{tabular}{|l|c|c|}
\hline
\textbf{Case} & \hspace{3pt} $\psi_x$ \hspace{3pt} & \hspace{3pt} $\alpha_{x, y}$ \hspace{3pt} \\
\hline
\text{x is a repeat} & \hspace{3pt} $\bot$ \hspace{3pt} & \hspace{3pt} $\top?$ \hspace{3pt} \\
\text{x is an exit} & \hspace{3pt} $\theta_{\Qlab_x}$ \hspace{3pt} & \hspace{3pt} ${-}$ \hspace{3pt} \\
\text{x is a companion} & \hspace{3pt} $\ldiap[\alpha_{z, x}^{*}]\psi_z$ \hspace{3pt} & \hspace{3pt} $\alpha_{z, x}^*; \alpha_{z, y}$ \hspace{3pt} \\
\text{x is otherwise of type $1$} & \hspace{3pt} $\psi_{z}$ \hspace{3pt} & \hspace{3pt} $\alpha_{z, y}$ \hspace{3pt} \\
\text{x is of type $2$} & \hspace{3pt} $\ldiap[\theta_{\Qlab_{x}}?]\psi_{z}$ \hspace{3pt} & \hspace{3pt} $\theta_{\Qlab_{x}}?;\alpha_{z, y}$ \hspace{3pt} \\
\text{x is of type $3$, not modal} & \hspace{3pt} $\bigvee \{ \psi_{z} \mid x \qedge z \}$ \hspace{3pt} & \hspace{3pt} $\bigcup \{ \alpha_{z, y} \mid x \qedge z, y \in \K{z} \}$ \hspace{3pt} \\
\text{x is of type $3$, modal} & \hspace{3pt} $\ldiap[a] \psi_{z}$ \hspace{3pt} & \hspace{3pt} $a;\alpha_{z, y}$ \hspace{3pt} \\
\hline
\end{tabular}
\]

\noindent
Based on these expressions the \emph{pre-interpolant} $\iota_x$ of a node $x \in Q$ is defined as:
\[
\iota_x \;\isdef\; \psi_x \lor \bigvee_{y \in \K{x}} \ldiap[\alpha_{x, y}] q_y.
\]
\end{ourdefinition}

\noindent
Note that the programs $\alpha_{x,y}$ and formulas $\psi_x$ do not contain internal variables.

\begin{ourdefinition}\label{d:itp}
We define the interpolant $\theta_r$ of the root $r$ of the cluster $C$ as
\[
\theta_{r} \isdef \iota_{r_{\sfQ}}.
\]
\end{ourdefinition}


\section{Correctness of the interpolant}\label{sec.correctness}

To prove Lemma~\ref{l.clusterInterpolation} and thereby establish the Craig interpolation property for \CPDL, we verify that the formula $\theta_r$ from Definition~\ref{d:itp} satisfies the three conditions of an interpolant for $\Gamma \| \Xi$:
first, the vocabulary condition $\Voc(\theta_r) \subseteq \Voc(\Gamma) \cap \Voc(\Xi)$ in \Cref{l.app.itpvoc}; 
second, that $\Gamma \| \theta_r$ is unsatisfiable in \Cref{l.app.GammaInterpolant}; 
and third, that $\mybar{\theta_r} \| \Xi$ is unsatisfiable in \Cref{l.InterpolantDelta}.

To establish the base cases of \Cref{l.app.itpvoc,l.app.GammaInterpolant,l.InterpolantDelta}, we first state the following auxiliary lemma that follows from the assumptions of \Cref{l.clusterInterpolation}.

\begin{lemma}\label{l.app.thetaProps}
For any sequent $\Delta \in \FRCp$ the following hold:
\begin{enumerate}
    \item $\proves \mybar{\theta_{\Delta}} \| \Delta$.
    \item $\proves \Plab^{l}_{t} \| \theta_{\Delta}$, for all $t \in C^{+}_{\Delta} \setminus C_{\Delta}$.
    \item $\Voc(\theta_\Delta) \subseteq \Voc(\Gamma) \cap \Voc(\Xi)$. 
\end{enumerate}
\end{lemma}

\subsection{Proof of vocabulary condition}

The following lemma is used to show the vocabulary condition.

\begin{lemma}\label{l:vocIsPreserved}
    For all $t \in C$ we have that:
    \[
    \Voc(\Plab^{l}_t) \subseteq \Voc(\Gamma)
    \quad \text{and} \quad
    \Voc(\Plab^{r}_t) \subseteq \Voc(\Xi).
    \]
\end{lemma}
\begin{proof}
    By root-to-leaf induction on $C$ using the fact that all our proof rules are analytic.
\end{proof}

We now state and show the vocabulary condition.

\begin{lemma}\label{l.app.itpvoc}
    For all nodes $x \in Q$, we have
    \[
    \Voc(\iota_{x}) \subseteq \big( \Voc(\Gamma) \cap \Voc(\Xi) \big)
    \cup \{ q_{y} \mid y \in \K{x} \}.
    \]
    As an immediate corollary, we have:
    $\Voc(\theta_{r}) \subseteq \Voc(\Gamma) \cap \Voc(\Xi)$.
\end{lemma}
\begin{proof}
The statement of the lemma is proved by a leaf-to-root induction on the 
structure of the quasi-proof $\sfQ$.

\begin{description}
\item[Case $k(x) = 1$, $x$ is a repeat.]
We have $\iota_{x} = q_{c(x)}$, so the claim follows from the fact that $c(x) \in \K{x}$.

\item[Case $k(x) = 1$, $x$ is an exit.]
We have $\iota_{x} = \theta_{\Qlab_x}$, so the claim holds by 
\Cref{l.app.thetaProps}.

\item[Case $k(x) = 1$, $x$ is a companion.]
In this case $x$ has a unique child $z$, and
\[
\iota_{x} = \bigvee_{y \in \K{x}} \ldiap[\alpha_{z, x}^*;\alpha_{z, y}]q_y \lor \ldiap[\alpha_{z, x}^*]\psi_z.
\]
and since $\K{z} = \K{x} \cup \{x\}$ we have that:
\[
\iota_{z} = \ldiap[\alpha_{z, x}] q_x \lor \bigvee_{y \in \K{x}} \ldiap[\alpha_{z, y}] q_y \lor \psi_z.
\]

By the inductive hypothesis and the definition of $\iota_{x}$, it is straightforward to calculate that:
    \[
    \Voc(\iota_x) \subseteq
    \big( \Voc(\Gamma) \cap \Voc(\Xi) \big)
    \cup 
    \big( \{ q_{y} \mid y \in \K{z} \} \setminus \{ q_x \}\big).
    \]

\item[Case $k(x) = 1$, $x$ is neither a leaf nor a companion.]
    In this case, $x$ has a unique child $z$, and $\iota_{x} = \iota_{z}$.
    By \Cref{l.KandLProps}, we have $\K{x} = \K{z}$.
    By the induction hypothesis for $z$ the claim follows.

\item[Case $k(x) = 2$.]
    In this case $x$ has a unique child $z$, and
    \[
    \iota_{x} = \ldiap[\theta_{\Qlab_x}?]\iota_{z}.
    \]
    By the induction hypothesis,
    \[
    \Voc(\iota_{z}) \subseteq \big( \Voc(\Gamma) \cap \Voc(\Xi) \big) \cup \{ q_{y} \mid y \in \K{z} \}.
    \]
Notice that $x$ is not a companion by \Cref{l.KandLProps}, and thus by the same lemma we get that $\K{z} = \K{x}$.
Also, by \Cref{l.app.thetaProps}, we find $\Voc(\theta_{\Qlab_x}) \subseteq \Voc(\Gamma) \cap \Voc(\Xi)$.
The claim follows from this.

\item[Case $k(x) = 3$, $\Qlab_{x}$ not modal.]
Here we have:
\[\iota_{x} = 
\bigvee_{x \qedge z} \psi_z \lor \bigvee_{y \in \K{x}} \ldiap[\bigcup\{\alpha_{z, y} \| x \qedge z, y \in \K{z}\}]q_y,
\]
where $\bigcup\{\alpha_{z, y} \mid x \qedge z, y \in \K{z}\}$ represents the choice program formed by combining all the programs $\alpha_{z, y}$ such that $x \qedge z$ and $y \in \K{z}$.
By the induction hypothesis, $\Voc(\iota_{z})$ is a subset of 
\[
\big( \Voc(\Gamma) \cap \Voc(\Xi) \big) \cup \{ q_{y} \mid y \in \K{z} \}.
\]
Since $x$ is not a companion by \Cref{l.KandLProps}, by the same lemma we find $\K{x} = \bigcup_{x \qedge z} \K{z}$.
    The claim follows from this.

\item[Case $k(x) = 3$, $\Qlab_{x}$ modal.]
In this case, $x$ has a unique child $z$, and $\iota_{x} = \ldiap[a]\iota_{z}$ where $a$ is the leading atomic program of the formula in focus in $\Qlab_{x}$. 
By the induction hypothesis,
    \[
    \Voc(\iota_{z}) \subseteq \big( \Voc(\Gamma) \cap \Voc(\Xi) \big) \cup \{ q_{y} \mid y \in \K{z} \}.
    \]
Since $x$ is not a companion by \Cref{l.KandLProps}, the same lemma gives $\K{z} =  \K{x}$.
Clearly $a \in \Voc(\Qlab_x) \subseteq \Voc(\Xi)$.
It remains to show that also $a \in \Voc(\Gamma)$. Recall that an action $a$ is in the vocabulary of $\Gamma$ if $a$ or $\conv{a}$ occur in $\Gamma$.
Arguing towards a contradiction assume that $a \notin \Voc(\Gamma)$.
Then by \Cref{l:vocIsPreserved} we get that $a \notin \Voc(\Plab^l(t))$ for any $t \in C$.
At any node $t \in R_x$ the rule $\RuDiaP^{r}$ is applied. Let $t \in R_x$ with child $s$. Then $\Plab^l_s = \ldiap[\conv{a}] \Plab^l_t$. But $\ldiap[\conv{a}] \Plab^l_t$ only contains formulas $\ldiap[\conv{a}]\chi$, such that $\ldiap[\conv{a}]\chi \in \FLN(\Plab^l_t)$ and if $\conv{a} \notin \Voc(\Plab^l_t)$ this implies that $\Plab^l_s$ is empty. This contradicts \Cref{l:itp4}.
\end{description}

Finally, the corollary that $\Voc(\theta_{r}) \subseteq \Voc(\Gamma) \cap \Voc(\Xi)$ is immediate since $\theta_r = \iota_{r_{\sfQ}}$ and $\K{r_{\sfQ}} = \emptyset$.
\end{proof}

\subsection{Proof of second condition: $\Gamma \| \theta_r$ is unsatisfiable}

In the proof of $\proves \Gamma \| \theta_r$ we will need the following definition and lemma.

\begin{ourdefinition}
Let $\pi$ be some split proof, possibly with assumptions.
We call $\pi$ right-focused if the path from the root of $\pi$ to any of its assumptions is right-focused (that is, the right component of each node on such a path is focused).
\end{ourdefinition}

For the proof of the following lemma, note that assumption-free proofs are automatically right-focused, and that $\rho_{\Sigma}: \calA \proves \Sigma \| \chi^{u}$ can only be right-focused if $\calA = \emptyset$.

\begin{lemma}
\label{l:ASaq}
Let $\calS$ and $\calA$ be finite sets of respectively (unfocused) sequents and split sequents, let $\alpha$ be some program, $q$ a proposition not occurring in either $\calS, \calA$ or $\alpha$, and $\chi$ some formula.
Assume that for every $\Sigma \in \calS$ there are right-focused proofs
$\pi_{\Sigma}: \{ \Pi \| q^{f} : \Pi \in \calS \} \proves \Sigma \| \ldiap[\alpha] q^{f}$ and $\rho_{\Sigma}: \calA \proves \Sigma \| \chi^{o}$.
Then we may construct a right-focused proof witnessing that $\calA \proves \Sigma \| \ldiap[\alpha^{*}] \chi^{o}$.
\end{lemma}
\begin{proof}
We first consider the case where $o = f$.
Abbreviate $\calS' \isdef \{ \Pi \| \ldiap[\alpha^{*}]\chi^{f} : \Pi \in \calS 
\}$, then for any $\Sigma \| \ldiap[\alpha^{*}]\chi^{f}$ in $\calS'$, we may 
consider the following derivation $\pi'_{\Sigma}: \calA \cup \calS' \proves 
\Sigma \| \ldiap[\alpha^{*}]\chi^{f}$:
\begin{center}
$\pi'_{\Sigma} \isdef $ \quad\quad
\begin{prooftree}
	\hypo{\calA
	}
    \ellipsis{$\rho_{\Sigma}$}{\Sigma \| \chi^{f}
    }
	\hypo{\calS'
	}
	\ellipsis{$\pi_{\Sigma}[\ldiap[\alpha^{*}]\chi^{f}]$}{
	    \Sigma \| \ldiap \ldiap[\alpha^*]\chi^{f}
    }        
	\infer[left label=$\RuStarD^{r}$]2{\Sigma \| \ldiap[\alpha^*]\chi^{f}
            }
\end{prooftree}
\end{center}
Here $\pi_{\Sigma}[\ldiap[\alpha^{*}]\chi^{f}]$ is the derivation we obtain from $\pi_{\Sigma}$ by substituting every occurrence of the proposition $q$ with the formula $\ldiap[\alpha^{*}]\chi^{f}$.
It is straightforward to check that $\pi_{\Sigma}[\ldiap[\alpha^{*}]\chi^{f}]$ is right-focused.
Note that $\ldiap[\alpha^*]\chi^{f}$, the formula in focus at the root of $\pi_{\Sigma}$ is actually principal there (\dag).

The main claim in the proof is that for all $n \geq 1$ and for every $\Sigma \| \ldiap[\alpha^{*}]\chi^{f}$ in $\calS'$ there is a right-focused proof $\pi_{n}: \calA \cup \calS' \proves \Sigma \| \ldiap[\alpha^*]\chi^{f}$, such that every open leaf $\ell$ which is labelled with some $\calS'$ assumption is a repeat, or else there are at least $n$ distinct split sequents from $\calS'$ on the path from the root of $\pi_n$ to $\ell$.

We prove this claim by induction on $n$.
In the base step, where $n = 1$, we can simply take the one-node proof of the split sequent $\Sigma \| \ldiap[\alpha^{*}]\chi^{f}$.
In the inductive step we assume a proof $\pi_{k}$ satisfying the above  constraints for $n=k$.
Now consider an arbitrary non-repeat leaf $\ell$ of $\pi_{n}$ which is labelled  with some split sequent in $\calS'$, say, with $\Sigma_{\ell} \| \ldiap[\alpha^{*}]\chi^{f}$.
If we replace each such $\ell$ with the derivation $\pi'_{\Sigma_{\ell}}$, it is easily verified that the resulting derivation $\pi_{k+1}$ satisfies the constraints for $n = k+1$.

Finally then we consider the derivation $\pi_{N}$, with $N{-}1$ being the number of all sequents in $\calS'$.
Obviously then, every open leaf in $\pi_{N}$ that is labeled with a split sequent in $\calS'$ must be a repeat.
It is then easy to see that $\pi_{N}$ is a proof: the success condition for the newly created repeat paths easily follows from the fact that $\pi_{N}$ is a right-focused proof, together with the observation (\dag).
This finishes the proof for the case where $o = f$.

For the case where $o = u$ we first note that this implies $\calA = \emptyset$ by definition of a right-focused proof.
Because $\Sigma$ is unfocused and $\pi_{\Sigma} \proves \Sigma \| \chi^{u}$ we easily obtain proofs $\rho_{\Sigma}$ of $\Sigma \| \chi^f$. 
Each $\rho_{\Sigma}$ is without assumptions and hence right-focused. 
We can thus apply the focused-case of the Lemma (with $\calA = \emptyset$) and obtain proofs of $\Sigma \| \ldiap[\alpha^{*}] \chi^{f}$.
Applying one \RuF rule then concludes the proof.
\end{proof}

With this definition and lemma, we are now ready to establish the second condition of the interpolant $\theta_r$ that states that $\Gamma \| \theta_r$ is unsatisfiable. In particular, we establish this by providing a $\SCPDLf$-proof of $\Gamma \| \theta_r$.

\begin{lemma}\label{l.app.GammaInterpolant}
    $\proves \Gamma \| \theta_r$.
\end{lemma}

\begin{proof}
The intuition underlying the proof is to show, by means of a leaf-to-root induction, that for each node $x$ of $\sfQ$ and for each $t \in R_{x}$ we can find a right-focused proof $\pi: \calC_x \proves \Plab_t^{l} \| \iota_x^f$, where $\calC_{x}$ is some suitable set of assumptions.
In the case where $x$ is a companion, the idea is to \emph{discharge} some of the assumptions, so that, when we arrive at the root $r_{\sfQ}$ of $\sfQ$ we obtain an assumption-free proof of the split sequent $\Lambda_t^{l} \| \iota_{r_{\sfQ}}^f = \Gamma \| \theta_r^f$.
For a proper proof-theoretic execution of this elimination procedure we need to prove a somewhat stronger claim, which involves separate statements on the constituting parts of the pre-interpolants.

\begin{claim}
For all $x \in Q$, for all $y \in \K{x}$ and all $t \in R_x$ we have that
(1) $\proves \Plab^{l}_t \| \psi_x^u$ and (2) there is a right-focused proof of $\Plab^{l}_t \| \ldiap[\alpha_{x, y}]q_y^f$ with assumptions $\calA_y \isdef \{(\Plab^{l}_s \| q_y^f) \mid s \in R_y\}$.
\end{claim}
\end{proof}

\begin{proof}[of Claim]
\begin{description}

We prove the Claim by a leaf-to-root induction on $x$.

\item[Case $k(x) = 1$, $x$ a repeat.]~%
    In this case, we have $\K{x} = \{ c(x) \}$, which implies $y = c(x)$ and 
    consequently $R_x = R_y$. 
    From this it follows that $t \in R_y$ and so we find that 
    $\Plab^{l}_t \| q_y^f \in \calA_y$.
    Note as well that in this case we have $\ldiap[\alpha_{x, y}]q_y = 
    \ldiap[\top?]q_y$ and $\psi_x = \bot$.
    It is then easy to show that $\proves \Plab^{l}_t \| \bot^u$, and to find a
    right-focused proof witnessing $\calA_y \proves \Plab^{l}_t \| 
    \ldiap[\top?]q_y^f$.

\item[Case $k(x) = 1$, $x$ an exit.]
    Since $\K{x} = \emptyset$ we only need to show that $\proves \Plab^{l}_t \| 
    \psi_x^u$.
    Observe that by definition we have $\psi_x = \theta_{\Qlab_x}$; but then by 
    Lemma~\ref{l.app.thetaProps} we obtain $\proves \Plab^{l}_t \| \theta_{\Qlab_x}^u$ for any $t \in R_{x}$ as required.

\item[Case $k(x) = 1$, $x$ a companion.]~%
    First we show (2).
    Notice than in this case $x$ has a unique child $z$ and $\K{z} = \K{x} \cup \{ x\}$ and $\alpha_{x, y} = \alpha_{z, x}^*;\alpha_{z, y}$.
    Let $y$ and $t$ be as in the claim, and note that $y \neq x$ since $x \notin \K{x}$ (Lemma~\ref{l.KandLProps}).
    By the induction hypothesis on $z$ we have right-focused proofs witnessing
    $
    \{ \Plab_{s}^{l}\| q_{x}^{f} : s \in R_{x} \} \proves 
    \Plab_{t}^{l}\| \ldiap[\alpha_{z,x}] q_{x}^{f}
    $
    and
    $
    \calA_{y} \proves \Plab^{l}\| \ldiap[\alpha_{z,y}] q_{y}^{f}.
    $
    Then we may apply Lemma~\ref{l:ASaq}, with $\calA = \calA_{y}$, $\calS = \{ \Plab_{s}^{l} : s \in R_{x} \}$, $\alpha = \alpha_{z,x}$, $q = q_{x}$ and $\chi = \ldiap[\alpha_{z,y}] q_{y}$.
    This yields a right-focused proof witnessing
    $
    \calA_{y} \proves 
    \Plab^{l}\| \ldiap[\alpha_{z,x}^{*}] \ldiap[\alpha_{z,y}] q_{y}^{f},
    $
    so that with one right application of $\RuConD$ we find that $\calA_{y} \proves \Plab^{l}\| \ldiap[\alpha_{z,x}^{*};\alpha_{z,y}] q_{y}^{f}$.
    This suffices, since we have $\alpha_{x, y} = \alpha_{z,x}^*;\alpha_{z, y}$.

    We now show (1). Recall that $\psi_x = \ldiap[\alpha_{z, x}^*]\psi_z$. Notice that in this case $R_x = R_z$. Write $R \isdef R_x$.

    Note that by the inductive hypothesis on $z$ there exists, for every $t \in R$, a proof of $\proves \Plab^{l}_t \| \psi_z^u$ and a right-focused proof of $\Plab_t^l\| \ldiap[\alpha_{z,x}]q_x^f$ with assumptions from the set $\{ (\Plab_s^l\| \ldiap[\alpha_{z,x}]q_x^f) \mid s \in R \}$.
    Then by Lemma~\ref{l:ASaq} with $\calS = \{ \Plab_{s}^{l} : s \in R_{z} \}$, $\calA = \emptyset$, $\alpha = \alpha_{z,x}$, $q = q_{x}$ and $\chi = \psi_z$ we get that $\proves \Plab^{l}_t \| \ldiap[\alpha_{z, x}]\psi_z^u$, as required.

\item[Case $k(x) = 1$, $x$ is neither a leaf nor a companion.]~
    Then $x$ has a unique child $z$ and we have that $\alpha_{x, y} \isdef \alpha_{z, y}$ and $\psi_{x} \isdef \psi_{z}$. Since $\calA_x = \calA_z$ and $R_x = R_z$, the claim is immediate by the inductive hypothesis.

\item[Case $k(x) = 2$.]~

    Then $x$ has a unique child $z$ and we have that $\alpha_{x, y} = \theta_{\Qlab_{z}}?;\alpha_{z, y}$ and $\psi_x = \ldiap[\theta_{\Qlab_{z}}?]\psi_{z}$.
    Since $\Qlab_x = \Qlab_z$, write $\Qlab \isdef \Qlab_z = \Qlab_x$, and since $\calA_x = \calA_z$, write $\calA \isdef \calA_x$.
    We will prove the case by establishing the following claim through a leaf-to-root inner induction on $t \in R_x$.

    \emph{Claim.} For all $t \in R_x$, we have that $\proves \Plab^{l}_t \| \ldiap[\theta_{\Qlab}?]\psi_{z}^u$ and there exists a right-focused proof of  $\calA\proves \Plab^{l}_t \| \ldiap[\theta_{\Qlab}?;\alpha_{z, y}]q_y^f$.

    We distinguish the following three subcases:
    \begin{description}
    \item[\it Subcase 1] 
        If $t \notin C$, then $t$ is an exit, meaning that $t \in C^{+} \setminus C$, and thus we have that $\proves \Plab^{l}_{t} \| \theta_{\Qlab}^u$ by \Cref{l.app.thetaProps}. From this, we can obtain the required proofs:

        \[
        \begin{prooftree}
            \hypo{\Plab^{l}_t \| \theta_{\Qlab}^u}
            \infer[left label=$\RuWeak^{r}$]1{\Plab^{l}_t \| \theta_{\Qlab}^u, \psi_z^u}
            \infer[left label=$\RuTestD^{r}$]1{\Plab^{l}_t \| \ldiap[\theta_{\Qlab}?]\psi_{z}^u}
        \end{prooftree}
        \quad\quad\quad
        \begin{prooftree}
            \hypo{\Plab^{l}_t \| \theta_{\Qlab}^u}
            \infer[left label=$\RuWeak^{r}$]1{\Plab^{l}_t \| \theta_{\Qlab}^u, \ldiap[\alpha_{z, y}]q_y^u}
            \infer[left label=$\RuU^{r}$]1{\Plab^{l}_t \| \theta_{\Qlab}^u, \ldiap[\alpha_{z, y}]q_y^f}
            \infer[left label=$\RuTestD^{r}$]1{\Plab^{l}_t \| \ldiap[\theta_{\Qlab}?]\ldiap[\alpha_{z, y}]q_y^f}
            \infer[left label=$\RuConD^{r}$]1{\Plab^{l}_t \| \ldiap[\theta_{\Qlab}?;\alpha_{z, y}]q_y^f}
        \end{prooftree}
        \]

        Notice that the proof of $\proves \Plab^{l}_t \| \ldiap[\theta_{\Qlab}?;\alpha_{z, y}]q_y^f$ is right-focused, as it does not have any assumptions.

    \item[\it Subcase 2] 
        If $t \in C^{r}_{\Qlab}$, then by the outer inductive hypothesis, we get that
        $\proves \Plab^{l}_t \| \psi_z^u$ and
        there exists a right-focused proof of
        $\calA \proves \Plab^{l}_t \| \ldiap[\alpha_{z, y}]q_y^f$.
        From this, we can construct the required proofs as follows:

        \begin{center}
        \begin{tabular}{c@{\hspace{1.5cm}}c}
        \begin{prooftree}
            \hypo{\Plab^{l}_t \| \psi_z^u}
            \infer[left label=$\RuWeak^{r}$]1{\Plab^{l}_t \| \theta_{\Qlab}^u, \psi_z^u}
            \infer[left label=$\RuTestD^{r}$]1{\Plab^{l}_t \| \ldiap[\theta_{\Qlab}?]\psi_{z}^u}
        \end{prooftree}
        &
        \begin{prooftree}
            \hypo{\calA}
            \ellipsis{}{\Plab^{l}_t \| \ldiap[\alpha_{z, y}]q_y^f}
            \infer[left label=$\RuWeak^{r}$]1{\Plab^{l}_t \| \theta_{\Qlab}^u, \ldiap[\alpha_{z, y}]q_y^f}
            \infer[left label=$\RuTestD^{r}$]1{\Plab^{l}_t \| \ldiap[\theta_{\Qlab}?]\ldiap[\alpha_{z, y}]q_y^f}
            \infer[left label=$\RuConD^{r}$]1{\Plab^{l}_t \| \ldiap[\theta_{\Qlab}?;\alpha_{z, y}]q_y^f}
        \end{prooftree}
        \end{tabular}
        \end{center}

        Notice that, since the proof of $\calA \proves \Plab^{l}_t \| \ldiap[\alpha_{z, y}]q_y^f$ is right-focused, the proof of
        $\calA \proves \Plab^{l}_t \| \ldiap[\theta_{\Qlab}?;\alpha_{z, y}]q_y^f$ is right-focused as well.

    \item[\it Subcase 3] 
        If $t \in C^{l}_{\Qlab}$, let $u_1, \dots, u_n$ list the children of $t$ in
        $C^+$.
        Since a left rule was applied, we have that $\Qlab_{u_i} = \Qlab$ for all
        $u_i \in C^{+}_{\Qlab} = R_x$. Hence, by the inner inductive hypothesis, it holds for $1 \leq i \leq n$ that
        \[
        \proves \Plab^{l}_{u_i} \| \ldiap[\theta_{\Qlab}?]\psi_{z}^u, \quad \text{and} \quad
        \calA \proves \Plab^{l}_{u_i} \| 
        \ldiap[\theta_{\Qlab}?;\alpha_{z, y}]q_y^f.
        \]
        By an application of the same left rule that was applied at $t$, the claim follows.
    \end{description}

    This finishes the proof of the Claim and, hence, that of the case.

\item[Case $k(x) = 3$, $\Qlab_{x}$ not modal.]
    In this case, $x$ has $n > 0$ children $z_{1}, \ldots, z_{n}$ in $\sfQ$.
    We need to show that $\proves \Plab^{l}_t \| \psi_{z_1} \lor \dots \lor \psi_{z_n}^u$, and that there exists a right-focused proof of
    \[
    \calA_y \proves \Plab^{l}_t \| \ldiap[\bigcup\{\alpha_{z, y} \| x \qedge z, y \in \K{z}\}]q_y^f.
    \]
    Here, recall that $\bigcup\{\alpha_{z, y} \mid x \qedge z, y \in \K{z}\}$ represents the choice program formed by combining all the programs $\alpha_{z, y}$ such that $x \qedge z$ and $y \in \K{z}$.

    By \Cref{l.uniformChildrenSimulation}, the children of $t$ can be listed as $t_1, \dots, t_n$ such that $t_i \in R_{z_i}$ for all $1 \leq i \leq n$. Since $\Plab^{l}_{t} = \Plab^{l}_{t_i}$, we will write $\Plab^{l} \isdef \Plab^{l}_t$.

    First, we will show that $\proves \Plab^{l} \| \psi_{z_1} \lor \dots \lor \psi_{z_n}^u$.
    By the inductive hypothesis, for each $i$, we have that $\proves \Plab^{l} \| \psi_{z_i}^u$. Using the $\RuOr^{r}$ rule repeatedly, we can construct the following proof:
    \[
    \begin{prooftree}
        \hypo{\Plab^{l} \| \psi_{z_1}^u}
        \hypo{\Plab^{l} \| \psi_{z_2}^u}
        \infer[left label=$\RuOr^{r}$]2{\Plab^{l} \| \psi_{z_1} \lor \psi_{z_2}^u}
        \ellipsis{}{\Plab^{l} \| \psi_{z_1} \lor \dots \lor \psi_{z_{n-1}}^u}
        \hypo{\Plab^{l} \| \psi_{z_n}^u}
        \infer[left label=$\RuOr^{r}$]2{\Plab^{l} \| \psi_{z_1} \lor \dots \lor \psi_{z_n}^u}
    \end{prooftree}
    \]

    Now we construct a right-focused proof of
    $\proves \Plab^{l} \| \ldiap[\bigcup\{\alpha_{z, y} \| x \qedge z, y \in \K{z}\}]q_y^f$

    Notice that for all children $z$ of $x$, whenever $y \in \K{z}$, there exists a right-focused proof $\pi_{z, y}$ of $\calA_{z}\proves \Plab^{l} \| \ldiap[\alpha_{z, y}]q_{y}^f$ by the inductive hypothesis.
    By a repeated application of $\RuChoiceD^{l}$ to every $\pi_{z, y}$ such that $y \in \K{z}$, we can obtain a proof $\rho$ of $\Plab^{l} \| \ldiap[\bigcup_{z \in \K{x}} \alpha_{z, y}]q_y^f$ with assumptions $\bigcup \{\calA_{z} \| x \qedge z, y \in \K{z}\} $. 

    By \Cref{l.KandLProps}, we know that $\bigcup \{\K{z} \| x \qedge z\} = \K{x}$.
    From this it is straightforward to see that $\bigcup \{\calA_{z} \| x \qedge z, y \in \K{z}\} = \calA_y$. Therefore we can conclude that $\rho$ is a proof of
    $\Plab^{l} \| \ldiap[\bigcup_{z \in \K{x}} \alpha_{z, y}]q_y^f$ with assumptions $\calA_y$.
    This completes the proof of this case.

\item[Case $k(x) = 3$, $\Qlab_{x}$ is modal.]
    Then $x$ has a unique child $z$ and we have $\alpha_{x, y} = a;\alpha_{z, y}$ 
    and $\psi_{x} = \ldiap[a]\psi_{z}$, for some atomic program $a$.
    By \Cref{l.KandLProps} we have that $\K{x} = \K{z}$, and thus $\calA_{x} = \calA_{z}$. Write $\calA \isdef \calA_{x}$.

    Let $t \in C^{r}_{\Qlab_x}$.
    By \Cref{l.uniformChildrenSimulation}, the unique child $t'$ of $t$ satisfies 
    $t' \in R_z$.
    By the inductive hypothesis we have $\calA \proves \Plab^{l}_{t'} \| 
    \ldiap[\alpha_{z, y}]q_y^f$ and $\proves \Plab^{l}_{t'} \| \psi_{z}^u$. 
    Then we can construct the required proofs as follows:
    \[
    \begin{prooftree}
        \hypo{\Plab^{l}_{t'} \| \psi_{z}^u}
        \infer[left label=$\RuU$]1{\Plab^{l}_{t'} \| \psi_{z}^f}  
        \infer[left label=$\RuDiaP^{r}$]1{\Plab^{l}_{t} \| \ldiap[a]\psi_{z}^f}
        \infer[left label=$\RuF$]1{\Plab^{l}_{t} \| \ldiap[a]\psi_{z}^u}    
    \end{prooftree}
    \quad\quad\quad
    \begin{prooftree}
        \hypo{\calA}
        \ellipsis{}{\Plab^{l}_{t'} \| \ldiap[\alpha_{z, y}]q_y^f}
        \infer[left label=$\RuDiaP^{r}$]1{\Plab^{l}_{t} \| \ldiap[a]\ldiap[\alpha_{z, y}]q_y^f}
        \infer[left label=$\RuConD^{r}$]1{\Plab^{l}_{t} \| \ldiap[a;\alpha_{z, y}]q_y^f}
    \end{prooftree}
    \]

    Observe that the proof of $\Plab^{l}_{t} \| \ldiap[a;\alpha_{z, y}]q_y^f$ is right-focused. This follows directly from the inductive hypothesis, as the proof of $\Plab^{l}_{t'} \| \ldiap[\alpha_{z, y}]q_y^f$ is right-focused, and the applied rules preserve that property.
\end{description}

\emph{This finishes the proof of the Claim.}

Now consider the root $r_{\sfQ}$ of $\sfQ$.
Since we have $\K{r_{\sfQ}} = \emptyset$ the Claim yields that $\proves \Plab^{l}_t \| \psi_{r_{\sfQ}}^u$ for all $t \in R_{r_{\sfQ}}$.
In particular, the root $r$ of the cluster $C$ belongs to $R_{r_{\sfQ}}$ and, since $\Plab^{l}_r = \Gamma$, we find that $\proves \Gamma \| \psi_{r_{\sfQ}}^u$.
Finally, unravelling the definitions we find that $\theta_{r} = \iota_{r_{\sfQ}}$, and, again since $\K{r_{\sfQ}} = \emptyset$, that $\iota_{r_{\sfQ}} = \psi_{r_{\sfQ}} \lor \bot$.
But then from $\proves \Gamma \| \psi_{r_{\sfQ}}^u$ we easily obtain $\proves \Gamma \| \theta_r$, as required.

\end{proof}

\subsection{Proof of third condition: $\mybar{\theta_r} \| \Xi$ is unsatisfiable}

The third interpolation condition states that the split sequent $\mybar{\theta_r} \| \Xi$ is unsatisfiable.
We will show this by providing an actual derivation as well, but here we use the unrestricted cut rule.
We let $\cproves$ denote derivability in the version of \SCPDLf where we allow the unrestricted cut rules $\RuCut^l$ and $\RuCut^r$.

Before providing the proof that $\cproves \mybar{\theta_r} \| \Xi$, we first state the following definition and lemma, with the purpose of simplifying the proofs.

\begin{ourdefinition}
\label{d:Qsh}
Let $x$ be a node in $\sfQ$, and let $\rho: \calA \cproves \Sigma^l \| \Sigma^r$ be a proof with assumptions.
We say that $\rho$ is $(\sfQ,x)$-shaped if $\calA = \{ (\mybar{q_y}^u \| \Qlab_y) \mid y \in \K{x} \}$, $\Sigma^r = \Qlab_{x}$ and for every open leaf $\ell$ of $\rho$ labeled with an assumption $\mybar{q_y}^{u} \| \Qlab_y$ there is a repeat $z$ in $\sfQ$ with $c(z) = y$, and such that the list of formulas in focus on the $\sfQ$-path from $x$ to $z$ is, up to repetitions, equal to the list of formulas in focus on the  path in $\rho$ from the root to $\ell$.
\end{ourdefinition}

\begin{lemma}\label{l:itp2aux1}
Let $\varphi$ and $\psi$ be equivalent \CPDL formulas.
Then we can transform any $(\sfQ,x)$-shaped proof $\rho: \calA \cproves \varphi^{u} \| \Delta$ into a $(\sfQ,x)$-shaped proof $\rho': \calA \cproves \psi^{u} \| \Delta$.
\end{lemma}

\begin{proof}
By completeness, there is a proof $\sigma$ of the split sequent $\mybar{\phi}^{u}, \psi^{u} \| \Delta$.
Using this, we construct the desired proof as follows:
\begin{center}
$\rho'$:\quad\quad
\begin{prooftree}
	\hypo{\calA
	}
	\ellipsis{$\rho$}{\phi^{u} \| \Delta}
	\hypo{}
	\ellipsis{$\sigma$}{\mybar{\phi}^{u}, \psi^{u} \| \Delta}
	\infer[left label=$\RuCut^{l}$]2{\psi^{u} \| \Delta}
\end{prooftree}
\end{center}
It is easy to verify that $\rho'$ is $(\sfQ,x)$-shaped if $\rho$ is so. This concludes the proof.
\end{proof}

\begin{lemma}\label{l.InterpolantDelta}
$\cproves \mybar{\theta_r} \| \Xi$.
\end{lemma}

\begin{proof}
By a leaf-to-root induction on $\sfQ$ we will prove the following claim, where we write $\calB_x \isdef \{ (\mybar{q_y}^u \| \Qlab_y) \mid y \in \K{x} \}$.

\begin{claim}
For every $x \in Q$ there is a $(\sfQ, x)$-shaped proof
$\pi_{x}: \calB_x \cproves \mybar{\iota_x}^{u} \| \Qlab_x$.
\end{claim}

\begin{proof}[of Claim]

    \begin{description}
    \item[Case $x$ is a repeat.]~%
    Here we have $\K{x} = \emptyset$ and thus $\calB_x = \{ (\mybar{q_{c(x)}}^{u} \|
    \Qlab_{x}) \}$.
    By definition, $\iota_x = \psi_x \lor \ldiap[\alpha_{x, c(x)}]q_{c(x)}$, where 
    $\alpha_{x, c(x)} = \top?$ and $\psi_x = \bot$. 
    It is then straightforward to find a (cut-free) $(\sfQ,x)$-shaped proof 
    $\pi_{x}: \calB_x \proves \mybar{\iota_x}^u \| \Qlab_x$.
    
    \item[Case $x$ is an exit.]~
    By definition we have $\iota_x = \theta_{\Qlab_x}$ and by 
    \Cref{l.app.thetaProps} that $\proves \mybar{\theta_{\Qlab_x}}^u \| \Qlab_x$.
    
    \item[Case $x$ is a companion.]~%
    
    Here $x$ has a unique child $z$, for which we have $\Qlab_x = \Qlab_z$; we will write $\Qlab \isdef \Qlab_x$.
    Furthermore recall that by the definition of pre-interpolants we have $\iota_{x} \equiv \iota_{z}[\iota_{x}/q_{x}]$.
    By completeness this means that there is some proof ${\rho:} \proves \mybar{\iota_{x}}^{u}, \iota_{z}[\iota_{x}/q_{x}]^{u} \| \Psi$.
    
    By the inductive hypothesis we have a $(\sfQ, z)$-shaped proof $\pi_{z}: \calB_x \cup \{ \mybar{q_{x}}^{u} \| \Qlab \} \cproves \mybar{\iota_z}^{u} \| \Qlab$.
    Substituting $q_x$ with $\iota_x$ everywhere in $\pi_{z}$ we obtain a proof $\pi_{z}[\iota_{x}]$ which we may cut with $\rho$ to obtain the following proof $\pi_x$ with assumptions:
    \begin{center}
    $\pi_x$ \quad\quad \begin{prooftree}
    	\hypo{\calB_x \cup \{ \mybar{\iota_{x}}^{u} \| \Qlab \}
    	}
        \ellipsis{$\pi_{z}[\iota_{x}]$}{\mybar{\iota_z}[\iota_{x}/q_{x}]^{u} \| \Qlab
        }
    	\hypo{ }
    	\ellipsis{$\rho$}{\iota_z[\iota_{x}/q_{x}]^{u},\mybar{\iota_{x}}^{u} \| \Qlab
        }        
    	\infer[left label=$\RuCut^{l}$]2{\mybar{\iota_{x}}^{u} \| \Qlab
    	}
    \end{prooftree}
    \end{center}
    Note that all paths from the root of $\pi_x$ to assumptions of the form $\mybar{\iota_{x}}^{u} \| \Qlab$ are successful, because $\pi_{z}$ is $(\sfQ, x)$-shaped and Lemma~\ref{l.KandLProps}\eqref{eqit:6}.
    Therefore all assumptions $\mybar{\iota_{x}}^{u} \| \Qlab$ in $\pi_x$ are discharged. This implies that the proof $\pi_{x}: \calB_x \cproves \mybar{\iota_x}^{u} \| \Qlab$ is $(\sfQ, x)$-shaped.
    
    \item[Case $k(x) = 1$, $x$ is neither a leaf nor a companion.]~
    In this case $x$ has a unique child $z$, for which we have $\alpha_{x, y} = \alpha_{z, y}$, $\psi_{x} = \psi_{z}$ and thus $\iota_{x} = \iota_{z}$. 
    Furthermore we have $\Qlab_{z} = \Qlab_{x}$ and $\calB_{z} = \calB_{x}$.
    Since $R_x = R_z$, the inductive hypothesis directly applies to $z$, providing us with a $(\sfQ,z)$-shaped proof $\pi_{z}: \calB_z \cproves \mybar{\iota_z}^{u} \| \Qlab_z$. 
    We may now simply take $\pi_{x} \isdef \pi_{z}$.
    
    \item[Case $k(x) = 2$.]~%
    In this case $x$ has a unique child $z$, for which we have $\Qlab_x = \Qlab_z$ and $\calB_x = \calB_z$.
    Write $\calB \isdef \calB_x$ and $\Qlab \isdef \Qlab_x$.
    By the inductive hypothesis we have a $(\sfQ,z)$-shaped proof $\pi_{z}: \calB \cproves \mybar{\iota_z}^u \| \Qlab$, and by \Cref{l.app.thetaProps} we have $\proves \mybar{\theta_{\Qlab}}^u \| \Qlab$.
    
    Applying the $\RuTestB^{l}$ rule we obtain a $(\sfQ,z)$-shaped proof witnessing that $\calB \cproves [\theta_{\Qlab}?]\mybar{\iota_z}^{u} \| \Delta$. 
    It is straightforward to verify that $[\theta_{\Qlab}?]\mybar{\iota_z} \equiv \mybar{\iota_x}$.
    But then by Lemma~\ref{l:itp2aux1} we obtain the desired proof $\pi_{x}: \calB \cproves \mybar{\iota_x}^{u} \| \Qlab$.
    
\item[Case $k(x) = 3$, $\Qlab_x$ is not modal.]~
In this case $x$ has $n > 0$ children $z_1, \dots, z_n$ in $\sfQ$.
Recall that by the uniformity of $\pi$, the same right rule \Ru is applied at each $t \in C_{\Psi_x}^{r} = R_x$.

By the inductive hypothesis, for each $z_i$, we have a $(\sfQ, z_i)$-shaped proof of $\calB_{z_i} \cproves \mybar{\iota_{z_i}}^{u} \| \Qlab_{z_i}$. 
    By repeated applications of the rules $\RuWeak^{l}$ and $\RuAnd^{l}$, we obtain, for each $z_i$, a $(\sfQ, z_i)$-shaped proof of $\pi_{z_i} : \calB_{z_i} \cproves \bigwedge_{1 \leq i \leq n}\mybar{\iota_{z_i}}^{u} \| \Qlab_{z_i}$.
    By an application of the rule \Ru we obtain a $(\sfQ, x)$-shaped proof \[\pi_x : \bigcup_{1 \leq i \leq n} \calB_{z_i} \cproves \bigwedge_{1 \leq i \leq n}\mybar{\iota_{z_i}}^{u} \| \Qlab_{x}\]
    By \Cref{l.KandLProps}, it follows that $\bigcup_{1 \leq i \leq n} \calB_{z_i} = \calB_x$ and thus, that $\pi_x : \calB_{x} \cproves \bigwedge_{1 \leq i \leq n}\mybar{\iota_{z_i}}^{u} \| \Qlab_{x}$.
    
    It is easy to see that $\bigwedge_{1 \leq i \leq n}\mybar{\iota_{z_i}} \equiv \mybar{\iota_x}$, and therefore Lemma~\ref{l:itp2aux1}, concludes the case.
    
    \item[Case $k(x) = 3$, $\Qlab_x$ is modal.]~
    In this case $x$ has a unique child $z$, for which we have $\iota_x \equiv \ldiap\iota_z$ and $\calB_x = \calB_z$. 
    Write $\calB = \calB_x$.
    By the inductive hypothesis there exists a $(\sfQ, z)$-shaped proof of $\calB \cproves \mybar{\iota_z}^{u} \| \Qlab_z$.
    By an application of the rule $\RuDia[a]^{r}$ to the formula in focus in $\Qlab_z$ we obtain a $(\sfQ, x)$-shaped proof $\calB \cproves [a]\mybar{\iota_z}^{u} \| \Qlab_x$.
    But since we have $[a]\mybar{\iota_z} \equiv \mybar{\iota_x}$, we may use Lemma~\ref{l:itp2aux1} to transform this proof into a $(\sfQ, x)$-shaped proof of $\calB \cproves \mybar{\iota_x}^{u} \| \Qlab_x$, as required.
\item[\it This finishes the proof of the Claim.]
\end{description}
\end{proof}

To finish the proof of Lemma~\ref{l.InterpolantDelta}, for the root $r_{\sfQ}$ of $\sfQ$ the Claim yields that $\calB_{r_{\sfQ}} \cproves 
\mybar{\iota_{r_{\sfQ}}}^{u} \| \Qlab_{r_{\sfQ}}$.
But since $\K{r_{\sfQ}} = \emptyset$, we find $\calB_{r_{\sfQ}} = \emptyset$, and as $\Qlab_{r_{\sfQ}} = \Xi$ and $\theta_r = \iota_{r_{\sfQ}}$, we may conclude that $\cproves \mybar{\theta_r} \| \Xi$, as required.
\end{proof}

\section{Conclusions}

We presented a sound and complete cyclic proof system for \CPDL and used it to show that the logic enjoys the Craig Interpolation Property. As a corollary we established that \CPDL also has the Beth Definability Property.

With some minor modifications we can make our approach work to prove interpolation for \PDL itself as well.
First, we simplify the rule \RuDiaP to the standard rule for (one-way) modal logic:

\begin{center}
\begin{prooftree}
\hypo{\phi^{\anColor{f}}, \Sigma}
\infer1{ \ldiap[a] \phi^{\anColor{f}},\lboxp[a] \Sigma,\Gamma}
    \end{prooftree}
\end{center}

We then verify that, with minor adaptations (that are in fact simplifications), the completeness proof still holds. 
Finally, we observe that the current definition of the interpolant will not involve the use of the converse modality, and check that the correctness proofs for the interpolant can be adapted to the system with the new modal rule.
We hope to supply the details of this approach in an expanded version of this paper.

Finally, it would be interesting to see whether one can find a proof system for \CPDL, such that the proof of correctness of the interpolant can be proven {inside} the system without the need of adding unrestricted cut. Achieving this seems to require adding admissible rules to $\CPDLf$.
Another open question is whether our method can be extended to other variants of \PDL such as \PDL with intersection \cite{Lutz2005} or deterministic \PDL \cite{BenAri1982}.

\bibliographystyle{splncs04}
\bibliography{bib.CPDL}

\begin{thebibliography}{10}
\providecommand{\url}[1]{\texttt{#1}}
\providecommand{\urlprefix}{URL }
\providecommand{\doi}[1]{https://doi.org/#1}

\bibitem{Afshari2019}
Afshari, B., Leigh, G.E.: Lyndon {Interpolation} for {Modal} $\mu$-{Calculus}.
  In: Language, {Logic}, and {Computation}: 13th {International} {Tbilisi}
  {Symposium}, {TbiLLC} 2019, {Revised} {Selected} {Papers} (2019).
  \doi{10.1007/978-3-030-98479-3_10}

\bibitem{baad:desc07}
Baader, F., Lutz, C.: Description logics. In: Blackburn, P., van Benthem, J.,
  Wolter, F. (eds.) Handbook of Modal Logic, Studies in logic and practical
  reasoning, vol.~3, pp. 757--820. North-Holland (2007).
  \doi{10.1016/S1570-2464(07)80015-2}

\bibitem{BenAri1982}
Ben-Ari, M., Halpern, J.Y., Pnueli, A.: Deterministic propositional dynamic
  logic: {Finite} models, complexity, and completeness. Journal of Computer and
  System Sciences  \textbf{25}(3),  402--417 (1982).
  \doi{10.1016/0022-0000(82)90018-6}

\bibitem{borz:tabl88}
Borzechowski, M.: Tableau-Kalk\"{u}l f\"{u}r PDL und Interpolation. Master's
  thesis, Department of Mathematics, Freie Universit\"{a}t Berlin (1988)

\bibitem{borz:prop25}
Borzechowski, M., Gattinger, M., Hansen, H.H., Ramanayake, R., Dalmas, V.T.,
  Venema, Y.: Propositional dynamic logic has {C}raig interpolation: a
  tableau-based proof (2025), \url{https://arxiv.org/abs/2503.13276}

\bibitem{cate:beth13}
ten Cate, B., Franconi, E., Seylan, I.: Beth definability in expressive
  description logics. Journal of Artificial Intelligence Research  \textbf{48},
   347--414 (2013). \doi{10.1613/JAIR.4057}

\bibitem{Emerson99}
Emerson, E., Jutla, C.: The complexity of tree automata and logics of programs.
  {SIAM} Journal of Computing  \textbf{29}(1),  132--158 (1999)

\bibitem{FL:PDL1979}
Fischer, M.J., Ladner, R.E.: Propositional dynamic logic of regular programs.
  Journal of Computer and System Sciences  \textbf{18}(2),  194--211 (1979).
  \doi{10.1016/0022-0000(79)90046-1}

\bibitem{gabb:inte05}
Gabbay, D.M., Maksimova, L.: Interpolation and Definability: modal and
  intuitionistic logics. Oxford University Press (2005)

\bibitem{Graedel2002}
Gr{\"a}del, E., Thomas, W., Wilke, T. (eds.): Automata, Logic, and Infinite
  Games, LNCS, vol.~2500. Springer (2002)

\bibitem{jung:sepa21}
Jung, J.C., Lutz, C., Pulcini, H., Wolter, F.: Separating data examples by
  description logic concepts with restricted signatures. In: Proceedings of the
  18th International Conference on Principles of Knowledge Representation and
  Reasoning, {KR} (2021). \doi{10.24963/KR.2021/37}

\bibitem{kloi:inte25}
Kloibhofer, J., Venema, Y.: Interpolation for the two-way modal $\mu$-calculus.
  In: 40th Annual ACM/IEEE Symposium on Logic in Computer Science (LICS 2025)
  (2025), to appear

\bibitem{Lutz2005}
Lutz, C.: {PDL} with {Intersection} and {Converse} {Is} {Decidable}. In: Ong,
  L. (ed.) Computer {Science} {Logic}. pp. 413--427. Springer, Berlin,
  Heidelberg (2005). \doi{10.1007/11538363_29}

\bibitem{mart:focu21}
Marti, J., Venema, Y.: Focus-style proof systems and interpolation for the
  alternation-free {\(\mu\)}-calculus. CoRR  \textbf{abs/2103.01671} (2021),
  \url{https://arxiv.org/abs/2103.01671}

\bibitem{mcmi:appl05}
McMillan, K.L.: Applications of {Craig} interpolants in model checking. In:
  Tools and Algorithms for the Construction and Analysis of Systems, 11th
  International Conference, {TACAS} 2005, Proceedings (2005).
  \doi{10.1007/978-3-540-31980-1_1}

\bibitem{Mostowski1991}
Mostowski, A.: Games with forbidden positions (1991), technical Report~78,
  Instytut Matematyki, Uniwersytet Gda\'{n}ski, Poland

\bibitem{pari:comp78}
Parikh, R.: The completeness of {Propositional Dynamic Logic}. In: Winkowski,
  J. (ed.) Mathematical Foundations of Computer Science 1978, Proceedings, 7th
  Symposium, Zakopane, Poland, September 4-8, 1978. Lecture Notes in Computer
  Science, vol.~64, pp. 403--415. Springer (1978).
  \doi{10.1007/3-540-08921-7_88}

\bibitem{prat:near80}
Pratt, V.: A near-optimal method for reasoning about action. Journal of
  Computer and System Sciences  \textbf{20},  231--254 (1980).
  \doi{10.1016/0022-0000(80)90061-6}

\bibitem{rena:inte08}
{Renardel de Lavalette}, G.R.: Interpolation in computing science: the
  semantics of modularization. Synthese  \textbf{164}(3),  437--450 (2008).
  \doi{10.1007/S11229-008-9358-Y}

\bibitem{sege:comp77}
Segerberg, K.: A completeness theorem in the modal logic of programs. Notices
  of the American Mathematical Society  \textbf{24}, ~522 (1977),
  \url{http://eudml.org/doc/209235}

\bibitem{Shamkanov2014}
Shamkanov, D.S.: Circular proofs for the {G}ödel-{L}öb provability logic.
  Mathematical Notes  \textbf{96}(3-4),  575--585 (sep 2014).
  \doi{10.1134/s0001434614090326}

\bibitem{troq:prop23}
Troquard, N., Balbiani, P.: Propositional dynamic logic. In: Zalta, E.N.,
  Nodelman, U. (eds.) The {Stanford} Encyclopedia of Philosophy. Metaphysics
  Research Lab, Stanford University, {F}all 2023 edn. (2023),
  \url{https://plato.stanford.edu/archives/fall2023/entries/logic-dynamic/}

\end{thebibliography}

\appendix

\appendix
\section{Parity games}\label{app.parityGames}
In this appendix we briefly define infinite two-player games, for more details we refer to \cite{Graedel2002}.
We fix two players that we shall refer to as $\eloi$ (Eloise, female) and $\abel$
(Abelard, male) and use $\Pi$ as a variable ranging over the set $\{ \eloi, \abel\}$.

A \emph{two-player game} is a quadruple $\bbG = (V,E,\Own,\WC)$ where $(V,E)$ is a graph, $\Own$ is a map $V \to \{ \eloi, \abel \}$, and $\WC$ is a set of infinite paths in $(V,E)$.
An \emph{initialised game} is a pair consisting of a game $\bbG$ and an element
$v$ of $V$, usually denoted as $\bbG@v$.

We will refer to $(V,E)$ as the \emph{board} of the game. 
Elements of $V$ will be called \emph{positions}, and $\Own(v)$ is the 
\emph{owner} of $v$.
Given a position $v$ for player $\Pi$, the set $E[v]$ 
denotes the set of \emph{moves} that are \emph{admissible
	for} $\Pi$ at $v$. We denote $V_{\Pi} \isdef \Own^{-1}(\Pi)$.
The set $\WC$ is called the \emph{winning condition} of the game.

A \emph{match} $\calM$ of the game $\bbG = (V,E,\Own,\WC)$
is a path through the graph $(V,E)$.
Such a match $\calM$ is \emph{full} if it is maximal as a path, that is, either
finite with\footnote{For a finite sequence $s = v_0...v_n$ we define $\first(s) \isdef v_0$ and $\last(s) \isdef v_n$.} $E[\last(\calM)] = \nada$, or infinite.
If a position has no $E$-successors, the owner of that positions \emph{gets stuck} and 
loses the match.
An infinite match is won by $\eloi$ if it belongs to the set $\WC$ and is won by $\abel$ otherwise.


Let $\PM_{\Pi}$ denote the collection of partial matches $\calM$ ending in a position $\last(\calM) \in V_{\Pi}$.
A \emph{strategy} for a player $\Pi$ is a partial function $f: \PM_{\Pi} \to V$ such that $f(\calM)
\in E[\last(\calM)]$ if $E[\last(\calM)] \neq \nada$.
A match $\calM = (v_{i})_{i<\kappa}$ is \emph{guided} by a $\Pi$-strategy $f$, in short $f$-guided, if $f(v_{0}v_{1}\cdots v_{n-1}) = v_{n}$ for all $n<\kappa$ 
such that $v_{0}\cdots v_{n-1}\in \PM_{\Pi}$.
%
A $\Pi$-strategy $f$ is \emph{winning for $\Pi$ from $v$} if $\Pi$ wins all $f$-guided full
matches initialised at $v$.
The game $\bbG$ is \emph{determined} if every position is winning for either 
$\eloi$ or $\abel$.

A strategy is \emph{positional} if it only depends on the last position of a 
partial match, namely, if $f(\calM) = f(\calM')$  whenever $\last(\calM) = \last(\calM')$;
such a strategy can and will be presented as a map $f: V_{\Pi} \to V$.

A \emph{parity game} is a board game $\bbG = (V,E,\Own,\WC_\Omega)$ in which the
winning condition $\WC_\Omega$ is given by a priority map $\Omega: V \to \Nat$ as follows: $\calM \in \WC_\Omega$ iff $\max\{\Omega(v) \| v \text{ occurs infinitely often in } \calM \}$ is even.
Such a parity game is usually denoted as $\bbG = (V,E,\Own,\Om)$.
The following theorem is independently due to Emerson \& Jutla~\cite{Emerson99}
and Mostowski~\cite{Mostowski1991}.

\begin{theorem}[Positional Determinacy]
	\label{thm.gamesPosDeterm}
	Let $\bbG = (G,E,\Own,\Om)$ be a parity game.
	Then $\bbG$ is determined, and both players have positional winning strategies.
\end{theorem}

\section{Proofs related to Section~\ref{s:cpdl}}\label{app.CPDL}

\subsection{Traces}

The \emph{length} $|e|$ of an expression $e$ is defined by a mutual induction on
formulas and programs.
Atomic formulas and programs have length one, and we set 
$|\phi\odot\psi| \isdef 1 + |\phi| + |\psi|$ if $\odot \in \{ \land,\lor \}$;
$|\alpha\odot\beta| \isdef 1 + |\alpha| + |\beta|$ if $\odot \in \{ \cup, {;} \}$;
$|\phi?| \isdef 1 + |\phi|$; $|\beta^{*}| \isdef 1 + |\beta|$.
The key clause in the definition is that we put $| \ldiap\phi| \isdef |\alpha| 
+ |\phi|$.

For a list $\vec{\delta}$ of programs, the formula $\edia{\vec{\delta}}{\psi}$ is inductively defined as follows: $\edia{\epsilon}{\psi} \isdef \psi$, and $\edia{\gamma\vec{\delta}}{\psi} \isdef \ldiap[\gamma] \edia{\vec{\delta}}{\psi}$.

\setcounter{theoremRem}{\thetheorem}
\setcounter{theorem}{\getrefnumber{p:inftr}}
\addtocounter{theorem}{-1}
\begin{proposition}
	Let $t = (\phi_{n})_{n<\omega}$ be an infinite trace.
	Then infinitely many $\phi_{n}$ are fixpoint formulas, and either cofinitely many $\phi_{n}$ are diamond formulas, or cofinitely many $\phi_{n}$ are box formulas.
\end{proposition}
\setcounter{theorem}{\thetheoremRem}
\begin{proof}
Let $\Inf(t)$ denote the set of formulas that occur infinitely often on $t$, and 
let $\phi \in \Inf(t)$ be of minimal length.
It is obvious that $\phi$ must be a fixpoint formula, since in all other cases 
the direct derivatives of $\phi$ are shorter than $\phi$ itself.

We only consider the case where $\phi$ is a diamond fixpoint formula, say, $\phi = \ldiap[\alpha^*]\psi$.
(The proof in the case where $\phi = \lboxp[\alpha^*]\psi$ is completely analogous.)
Let $k$ be such that $\Inf(t) = \{ \phi_n \mid n \geq k \}$, and such thtat $\phi_k = \phi$.
The proposition then follows from the claim below.
\medskip

\textit{Claim}
For all $n \geq k$ the formula $\phi_n$ is of the form $\phi_n = 
\edia{\vec{\delta}}{\phi}$, for some list $\vec{\delta} = \delta_1\cdots \delta_m$ of 
programs, where each $\delta_i$ is shorter than $\alpha$.
\medskip

The claim can be proved by a straightforward induction on $n$.
In the base case, where $n = k$, we have $\phi_n = \phi = \edia{\epsilon}{\phi}$.

For the induction step we assume as our induction hypothesis that $\phi_{n} = 
\edia{\vec{\delta}}{\phi}$ for some program list $\vec{\delta}$, and we make a 
case distinction.
In case $\vec{\delta} = \epsilon$ we have $\phi_{n} = \phi = 
\ldiap[\alpha^*]\psi$, so that $\phi_{n+1} \in \{ \psi, 
\ldiap\ldiap[\alpha^{*}]\psi \}$.
However by the assumption on $k$ the formula $\phi_{n+1}$ cannot be shorter than
$\phi$, so that we find $\phi_{n+1} = \ldiap\ldiap[\alpha^*]\psi =
\edia{\alpha}{\phi}$.

In case $\vec{\delta} \neq \epsilon$ we may write $\vec{\delta} = \beta\vec{\gamma}$ and we make a further case distinction as to the nature of $\beta$.
\begin{description}
\item[\it Case $\beta = a$ for some $a \in Act$.]
Here we find $\phi_{n+1} = \edia{\vec{\gamma}}{\phi}$.
\item[\it Case $\beta = \tau?$.]
We find that $\phi_{n+1} \in  \{ \tau, \edia{\vec{\gamma}}{\phi} \}$, but since (by assumption on $k$) $\phi_{n+1}$ cannot be shorter than $\phi$ only the second option is possible.
\item[\it Case $\beta = \beta_0;\beta_1$.]
We obtain $\phi_{n+1} = \edia{\beta_0\beta_1\vec{\gamma}}{\phi}$.
\item[\it Case $\beta = \beta_0\cup\beta_1$.]
We obtain $\phi_{n+1} = \edia{\beta_0\vec{\gamma}}{\phi}$
or $\phi_{n+1} = \edia{\beta_1\vec{\gamma}}{\phi}$.
\item[\it Case $\beta = \beta_0^*$.]
We find that either 
$\phi_{n+1} = \edia{\vec{\gamma}}{\phi} $ or 
$\phi_{n+1} = \edia{\beta_0\beta_0^*\vec{\gamma}}{\phi} $.
\end{description}
In all cases it is straightforward to verify that $\phi_{n+1}$ has the required shape.

This finishes the proof of the Claim and, hence, that of the Proposition.
\end{proof}

\subsection{Adequacy of the game semantics}

In this section we present the standard, compositional semantics of \CPDL, 
and we show its equivalence to the game semantics of Definition~\ref{d:gsem}.
In the compositional semantics, given a Kripke model $\bbS$, formulas and 
programs are inductively interpreted as, respectively, subsets of and 
binary relations over the carrier of $\bbS$.
For the modal clause of this definition we need the fact that any binary relation
$R$ on a set $S$ induces two operations on $\powset(S)$:
\[\begin{array}{lll}
   \exop{R}(U) & \isdef &
   \{ s \in S \mid R[s] \cap U \neq \nada \}
\\ \unop{R}(U) & \isdef & 
  \{ s \in S \mid R[s] \subseteq U \}.
\end{array}\]

\begin{ourdefinition}
Given a Kripke model $\bbS = (S,R,V)$, by a mutual recursion on formulas and programs we define the meaning $\mng{\phi}^{\bbS} \subseteq S$ of a formula $\phi$ in $\bbS$:
\[
\begin{array}{lllclll}
\mngS{\top} & \isdef & S & \; &
 \mngS{\bot} & \isdef & \nada 
\\ \mngS{p} & \isdef & V(p) 
  && \mngS{\mybar{p}} & \isdef & S \setminus V(p) 
\\ \mngS{\phi\land\psi} & \isdef & \mngS{\phi}\cap \mngS{\psi}
  && \mngS{\phi\lor\psi} & \isdef & \mngS{\phi}\cup \mngS{\psi}
\\ \mngS{\ldiap\phi} & \isdef & \exop{R_{\alpha}}(\mngS{\phi})
  && \mngS{\lboxp\phi} & \isdef & \unop{R_{\alpha}}(\mngS{\phi})
\end{array}
\]
and we extend the maps $R= \{R_a \| a \in \Act\}$ to provide an accessibility relation $R_{\alpha} \subseteq S \times S$ to an arbitrary program $\alpha$:
\[
\begin{array}{llll}
R_{\alpha;\beta} & \isdef & R_{\alpha} ; R_{\beta}
  & (= \{ (s,u) \mid (s,t) \in R_{\alpha} \text { and } (t,u) \in 
   R_{\beta} \text { for some } t \})
\\ R_{\alpha\cup\beta} & \isdef & R_{\alpha} \cup R_{\beta}
\\ R_{\alpha^{*}} & \isdef & R_{\alpha}^{*} 
  & (= \bigcup_{n} R_{\alpha}^{n})
\\ R_{\phi?} & \isdef & \multicolumn{2}{l}{\{ (s,s) \mid s \in \mngS{\phi} \}}
\end{array}
\]
\end{ourdefinition}

In the evaluation game $\EG(\bbS)$ we let $\Win_{\eloi}(\bbS)$ denote the
winning positions for $\eloi$, and we write $\gmngS{\phi} \isdef \{ s \in S 
\mid (\phi,s) \in \Win_{\eloi}(\bbS)\}$.

\begin{theorem}[Adequacy]
For every formula $\phi$ we have
\begin{equation}
\label{eq:adeqf}
\mngS{\phi} = \gmngS{\phi}, \text{ for every model } \bbS.
\end{equation}
\end{theorem}

\begin{proof}
By a mutual induction on formulas and programs we will show that every formula
$\phi$ satisfies \eqref{eq:adeqf}, while for every program $\alpha$ we have
\begin{equation}
\label{eq:adeqp}
\gmngS{\ldiap\psi} = \exop{R_\alpha^\bbS}(\gmngS{\psi}) 
\text{ and }
\gmngS{\lboxp\psi} = \unop{R_\alpha^\bbS}(\gmngS{\psi}),
   \text{ for every } \bbS \text{ and }\psi.
\end{equation}

The proof of \eqref{eq:adeqf} is routine, so we confine ourselves to a few
examples.
The case where $\phi$ is atomic is immediate by the definitions.
In the induction step where $\phi$ is a disjunction, say, $\phi = \phi_{0} \lor
\phi_{1}$, we reason as follows: $s \in \mngS{\phi} = 
\mngS{\phi_{0}}\cup \mngS{\phi_{1}}$ iff 
$s \in \mngS{\phi_{0}}$ or $s \in \mngS{\phi_{1}}$ iff (IH)
$s \in \gmngS{\phi_{0}}$ or $s \in \gmngS{\phi_{1}}$ iff 
$s \in \gmngS{\phi_{0} \lor \phi_{1}}$, where the last equivalence is 
based on an obvious game-theoretical observation.
For the case where $\phi = \ldiap\phi'$ we reason as follows.
By definition we have 
$\mngS{\ldiap\phi'} = \exop{R_{\alpha}^{\bbS}}(\mngS{\phi'})$,
and by applications of the induction hypothesis (for $\phi'$ and $\alpha$, 
respectively), we find that 
$\exop{R_{\alpha}^{\bbS}}(\mngS{\phi'}) = \exop{R_{\alpha}^{\bbS}}(\gmngS{\phi'}) 
= \gmngS{\ldiap\phi'})$.
Clearly then we have $\mngS{\ldiap\phi'} = \gmngS{\ldiap\phi'})$.
\medskip

For the proof of \eqref{eq:adeqp} we only cover the statement on diamond 
formulas, and we leave the cases where $\alpha$ is atomic or of the form 
$\beta\cup \gamma$ as exercises for the reader.

In the case where $\alpha$ is a test, say, $\alpha = \tau?$, it is easy to see
that $\gmngS{\ldiap[\tau?]\psi} = \gmngS{\tau} \cap \gmngS{\psi}$.
For the right hand side we have $s \in \exop{R_{\tau?}^\bbS}(\gmngS{\psi})$ 
iff there is a $t \in R_{\tau?}[s] \cap \gmngS{\psi}$ 
iff $s \in \mngS{\tau} \cap \gmngS{\psi}$
iff (by induction hypothesis on $\tau$) $s \in \gmngS{\tau} \cap \gmngS{\psi}$,
as required.

In the case where $\alpha$ is of the form $\alpha = \beta;\gamma$ we apply the
induction hypothesis to $\beta$ and $\gamma$, respectively, and find that
$\gmngS{\ldiap[\beta]\ldiap[\gamma]\psi} =
\exop{R_\beta^\bbS}(\gmngS{\ldiap[\gamma]\psi}) =
\exop{R_\beta^\bbS}\big(\exop{R_\gamma^\bbS}(\gmngS{\psi})\big)$.
But then we are done, since we obviously have that $\exop{R_{\beta;\gamma}^\bbS}$
is the composition of $\exop{R_\beta^\bbS}$ and 
$\exop{R_\gamma^\bbS}$.
\medskip

The key case in the proof is where $\alpha$ is an iteration, say, $\alpha = 
\beta^{*}$.
For the inclusion $\supseteq$ we observe that $\exop{R_{\beta^{*}}^\bbS}(A) 
= \bigcup_{n<\omega} \exop{R_{\beta}^{\bbS}}^{n}(A)$, for all 
$A \subseteq S$, so that it suffices to show that 
$\exop{R_{\beta}^{\bbS}}^{n}(\gmngS{\psi}) \subseteq 
\gmngS{\ldiap[\beta^{*}]\psi}$, for all $n$.
This inclusion we can establish by a straightforward inner induction on $n$.
In the base step we have 
$\exop{R_{\beta}^{\bbS}}^{0}(\gmngS{\psi}) = \gmngS{\psi}$,
and it is obvious that $\gmngS{\psi} \subseteq \gmngS{\ldiap[\beta^{*}]\psi})$.
In the inductive step we find
$\exop{R_{\beta}^{\bbS}}^{n+1}(\gmngS{\psi}) =
\exop{R_{\beta}^{\bbS}}\exop{R_{\beta}^{\bbS}}^{n}(\gmngS{\psi})$.
Respective applications of the inner induction hypothesis (on $n$) and the outer induction hypothesis (on $\beta$) show that
$\exop{R_{\beta}^{\bbS}}\exop{R_{\beta}^{\bbS}}^{n}(\gmngS{\psi}) \subseteq \exop{R_{\beta}^{\bbS}}(\gmngS{\ldiap[\beta^{*}]\psi})) \subseteq \gmngS{\ldiap[\beta]\ldiap[\beta^{*}]\psi}$.
Finally, it is obvious that 
$\gmngS{\ldiap[\beta]\ldiap[\beta^{*}]\psi} \subseteq
\gmngS{\ldiap[\beta^{*}]\psi})$, so that we are done.
\medskip

For the opposite inclusion $\subseteq$ of \eqref{eq:adeqp} in the case where 
$\alpha = \beta^{*}$ we have to do more work.
To reduce notational clutter we let $A$ denote the right hand side of 
the equation, that is, $A \isdef \exop{R_{\beta^{*}}^{\bbS}}(\gmngS{\psi})$;
it is an easy consequence of the definition of $R_{\beta^{*}}^{\bbS}$ that
\begin{equation}
\label{eq:adeq1}
A = \gmngS{\psi} \cup 
\exop{R_{\beta}^{\bbS}}A.
\end{equation}
We will show that $\gmngS{\ldiap[\beta^{*}\psi} \subseteq A$ by providing $\abel$ with 
a winning strategy in the game $\EG(\bbS)@(\ldiap[\beta^{*}]\psi,s)$, for an 
arbitrary state $s \not\in A$.

In order to define this strategy we use an auxiliary structure.
Take a fresh proposition letter $p$ and consider the model $\bbS_{A} \isdef \bbS[p \mapsto A]$; that is, we modify the valuation of $\bbS$ so that in $\bbS_{A}$ the proposition letter $p$ is interpreted as the set $A$.
Fix some winning positional strategy $g$ for $\abel$; that is, $g$ is winning for every position in $\Win_{\abel}(\EG(\bbS_{A}))$.
Furthermore, observe that we have 
\begin{equation}
\label{eq:adeq2}
\exop{R_{\beta}^{\bbS}}(A)
= \exop{R_{\beta}^{\bbS}}(\mng{p}^{\bbS_{A}})
= \exop{R_{\beta}^{\bbS_{A}}}(\gmng{p}^{\bbS_{A}})
= \gmng{\ldiap[\beta]p}^{\bbS_{A}},
\end{equation}
where we use the fact that $p$ does not occur in $\beta$ in the second equality, and the induction hypothesis on $\beta$ in the last one.
Then for any state $t$ in $S$ it follows from $t\not\in A$, \eqref{eq:adeq1} and \eqref{eq:adeq2} that $t \not\in  \gmng{\ldiap[\beta]p}^{\bbS_{A}}$, which by determinacy of $\EG(\bbS_{A})$ means that $g$ is winning for $\abel$ in $\EG(\bbS_{A})@(\ldiap[\beta]p,t)$:
\begin{equation}
\label{eq:adeq3}
t\not\in A \text{ implies $g$ is winning for $\abel$ in }
\EG(\bbS_{A})@(\ldiap[\beta]p,t). 
\end{equation}
Furthermore observe that since $\bbS_{A}$ is an expansion of $\bbS$, we may see $g$ as a strategy for $\EG(\bbS)$ as well.

We can now define $\abel$'s strategy $h$ in $\EG(\bbS)$ as follows:
\begin{itemize}
\item at a position of the form $\ldiap[\tau?]\chi$ play as follows:
\\ - pick $(\tau,u)$ if $(\tau,u) \in \Win_{\abel}(\EG(\bbS))$ and 
  continue with the strategy $g$;
\\ - otherwise, pick $\chi$.
\item at any other position $(\chi,u)$ play $g$ if $(\chi,u) \in 
\Win_{\abel}(\EG(\bbS))$; otherwise play randomly.
\end{itemize}

In order to show that $h$ is winning for $\abel$ in 
$\EG(\bbS)@(\ldiap[\beta^{*}]\psi,s)$, we need the following Claim.
\medskip

\noindent
\textit{Claim} 
Let $\pi = (\phi_{i},s_{i})_{i\leq n}$ be some partial $h$-guided match of 
$\EG(\bbS)@(\ldiap[\beta^{*}]\psi,s)$ where $s_{0} = s \not \in A$.
If at any position $(\ldiap[\beta^*]\psi,t)$ in $\pi$ Eloise picks $(\ldiap[\beta]\ldiap[\beta^*]\psi,t)$ and $\phi_{n} = \ldiap[\beta^{*}]\psi$, then $s_{n} \not\in A$.\\

\noindent
\textit{Proof of Claim} 
First of all note that at the start $(\phi_{0},s_{0}) = (\ldiap[\beta^{*}]\psi,s)$ of the match by assumption $\eloi$ picks $(\phi_{1},s_{1}) = (\ldiap[\beta]\ldiap[\beta^{*}]\psi,s)$, and that the position $(\ldiap[\beta]p,s) \in \Win_{\eloi}(\EG(\bbS_{A}))$.
Now let $m$ with $1 \leq m\leq n$ be minimal such that $\phi_{m} = \ldiap[\beta^{*}]\psi$. We first prove that $s_m \notin A$, inductively this implies that $s_n \notin A$ as intended.
In order to do so we claim that the match $\pi' = (\phi_{i},s_{i})_{1 \leq i \leq m}$ is 
of the form $\rho'[\ldiap[\beta^{*}]\psi/p]$ for some $g$-guided $\EG(\bbS_{A})$
match $\rho'$.


To see this, we show that for every $i$ with $1 \leq i \leq m$ there are program lists $\lambda_{i}$ such that (\dag) $\phi_{i} = \edia{\lambda_{i}}{\ldiap[\beta^{*}]\psi}$ for all $i$, while (\ddag) the sequence $(\edia{\lambda_{i}}{p},s_{i})_{1\leq i \leq m}$ is a $g$-guided partial $\EG(\bbS_{A})$-match.
This statement is obvious for $i = 1$, as  $\phi_{1} = \ldiap[\beta]\ldiap[\beta^*]\psi$ by assumption, meaning that $\lambda_{1} \isdef \beta$.
In the induction step we assume that we have defined the program lists 
$\lambda_{1}, \ldots, \lambda_{i}$ satisfying (\dag) and (\ddag) for some $1 \leq i < m$.
Since $m$ is minimal with $\phi_m = \ldiap[\beta^*]\psi$ the program list $\lambda_{i}$ is nonempty and so it must be of the form $\gamma \mu$ for some program $\gamma$ and program list $\mu$.

The only case of interest is where $\gamma$ is a test, say, $\gamma = \tau?$.
The position $(\phi_{i},s_{i})$ in $\pi'$ is $(\phi_{i},s_{i}) 
= ((\edia{\tau?\mu}{\ldiap[\beta^{*}]\psi},s_i)
= (\ldiap[\tau?]\edia{\mu}{\ldiap[\beta^{*}]\psi},s_i)$.
We claim that $\phi_{i+1} = \edia{\mu}{\ldiap[\beta^{*}]\psi}$.
To see this, note that we cannot have $\phi_{i+1} = \tau$, since all subsequent formulas in $\pi$ would have to belong to the closure of $\tau$, and this clearly does not hold for the formula $\phi_{m} = \ldiap[\beta^{*}]\psi$.
But since $\abel$ played according to his strategy $h$, this means that $(\tau,s_{i})$ is a winning position for $\eloi$, in both $\EG(\bbS)$ and $\EG(\bbS_{A})$.
Hence, by our assumption that $g$ is a winning strategy for $\abel$, in $\EG(\bbS_{A})$, at position $(\edia{\tau?\mu}{p},s_i)$, it will tell $\abel$ to move to position $(\edia{\mu}{p}$.
In other words, the new position in $\EG(\bbS)$ is $(\phi_{i+1},s_{i+1}) = (\edia{\mu}{\ldiap[\beta^{*}]\psi},s_{i+1})$ while the new position in $\EG(\bbS_{A})$ is $(\edia{\mu}{p},s_{i+1})$.
Obviously then, if we define $\lambda_{i+1} \isdef \mu$ the conditions (\dag) and (\ddag) hold for $i+1$, as required.

For the case $i = m$ the statements (\dag) and (\ddag) imply that $s_m \not\in A$.
This implies $s_n \notin A$ and we may consider the claim to be proved.
\medskip

Now consider an arbitrary $h$-guided full match $\pi = (\phi_i,s_i)_{i\leq \kappa}$ of \newline
$\EG(\bbS)@(\ldiap[\beta^{*}]\psi,s)$, where $s \not\in A$.
To see why $\pi$ must be won by $\abel$, we distinguish cases.

Let $n$ be maximal such that the assumptions of the Claim are satisfied, meaning that $\phi_{n} = \ldiap[\beta^{*}]\psi$ and Eloise picks $(\ldiap[\beta]\ldiap[\beta^*]\psi,t)$ at any position $(\ldiap[\beta^*]\psi,t)$ in $(\phi_i,s_i)_{i\leq n}$.

If $n$ is undefined, then there are infinitely many positions of the form $(\ldiap[\beta^{*}]\psi,t)$ in $\pi$, then by definition of the winning conditions, 
it constitutes a win for $\abel$.

Otherwise write $t \isdef s_{n}$.
It follows by the Claim that $t \not\in A$, and, hence, 
by \eqref{eq:adeq1} that $t \not\in \gmngS{\psi}$.
By the determinacy of the evaluation game this means that $(\psi,t) \in 
\Win_{\abel}(\EG(\bbS))$.
Hence if $\eloi$ picks position $(\phi_{n+1},s_{n+1}) = (\psi,t)$ then the 
remainder of $\pi$ will simply be guided by $\abel$'s winning strategy $g$, 
resulting in a win for $\abel$.
Now assume that $\eloi$ picks position $(\phi_{n+1},s_{n+1}) = 
(\ldiap[\beta]\ldiap[\beta^{*}]\psi,t)$.
In case $\abel$ picks a test formula at some position after stage $n+1$,
the remaining tail of $\pi$ is guided by his winning strategy $g$ and so he 
wins $\pi$.
But if $\abel$ never picks a test formula it means that for all $i > n$ the formula $\phi_{i}$ 
is the form $(\edia{\lambda}{\ldiap[\beta^{*}]\psi},t)$ for some nonempty list of programs $\lambda$. Since $\pi$ is a full match this can only be 
the case if $\pi$ is infinite, but then $\pi$, having a tail of diamond formulas,
is won by $\abel$.

In other words, we have proved that $h$ is winning for $\abel$ in 
$\EG(\bbS)@(\ldiap[\beta^{*}]\psi,s)$ indeed, and this suffices to prove the 
inclusion $\subseteq$ of \eqref{eq:adeqp} in the case where $\alpha = \beta^{*}$.
\end{proof}

\section{Proofs related to Section \ref{sec.splitCompletness}}\label{app.splitCompleteness}

\subsection{Soundness}

We prove the soundness of the split system \SCPDLf by first showing the soundness of \CPDLf and then translating \SCPDLf proofs of $\Gamma \| \Delta$ to \CPDLf proofs of $\Gamma, \Delta$.

The following lemma deals with the local soundness of our rules and can be 
proven straightforwardly:
\begin{lemma}\label{lem.soundLocal}
	Let 
	\[ 
	\begin{prooftree}
		\hypo{\Delta_1}
		\hypo{\cdots}
		\hypo{\Delta_n}
		\infer[left label =\Ru]3{\Gamma}
	\end{prooftree}
	\] 
	be a rule instance of Figure \ref{fig.rulesCPDL}. 
	If $\Gamma$ is satisfiable, then $\Delta_i$ is satisfiable for some $1 \leq i \leq n$. 
	
	More concretely, let $\bbS,s$ be a pointed model and $f$ a positional strategy for $\eloi$ in $\calE(\bbS)$ such that $\bbS,s \sat_f \Gamma$. 
	Then,
	\begin{enumerate}
		\item if $\Ru \neq \RuDia[a]$ then $\bbS,s \sat_f \Delta_i$ for some $1 \leq i \leq n$,
		\item if $\Ru = \RuOr$ with principal formula $\phi_0 \lor \phi_1$, then $\bbS,s \sat_f \phi_j, \Gamma$, where $f(\phi_0 \lor \phi_1,s) = (\phi_j,s)$,
		\item if $\Ru = \RuChoiceD$ with principal formula $\ldiap[\alpha_0 \cup \alpha_1]\phi$, then $\bbS,s \sat_f \ldiap[\alpha_j]\phi, \Gamma$, where $f(\ldiap[\alpha_0 \cup \alpha_1]\phi,s) = (\ldiap[\alpha_j]\phi,s)$,
		\item if $\Ru = \RuStarD$ with principal formula $\ldiap[\alpha^*]\phi$, then $\bbS,s \sat_f \ldiap[\alpha]\ldiap[\alpha^*]\phi, \Gamma$ if \newline $f(\ldiap[\alpha^*]\phi,s) = (\ldiap[\alpha]\ldiap[\alpha^*]\phi,s)$ and $\bbS,s \sat_f \phi, \Gamma$ else,
		\item if $\Ru = \RuTestB$ with principal formula $\lboxp[\psi?]\phi$, then $\bbS,s \sat_f \mybar{\psi}, \Gamma$ if $f(\lboxp[\psi?]\phi,s) = (\mybar{\psi},s)$ and $\bbS,s \sat_f \phi, \Gamma$ else,
		\item if $\Ru = \RuDia$ with principal formula $\ldiap \phi$ and $\Gamma = \ldiap \phi, \lboxp \Sigma, \Pi$, then $\bbS,t \sat_f \phi,\Sigma, 
		\ldiap[\conv{a}]\Pi$, where $f(\ldiap \phi,s) = (\phi,t)$. 
	\end{enumerate}
\end{lemma}

\begin{theorem}\label{thm.soundnessNormal}
	If $\proves \Gamma$, then $\Gamma$ is unsatisfiable.
\end{theorem}
\begin{proof}
	By contraposition we show that if $\Gamma$ is satisfiable, then Builder has a 
	winning strategy in $\calG := \calG(\Gamma)@\Gamma$. 
	So assume that there is a pointed model $\bbS,s$ and a positional strategy $f$ for $\eloi$ 
	in the game $\calE := \calE(\bbS)$ such that $\bbS,s \sat_f \Gamma$.
	We will  construct a winning strategy $\overline{f}$ for Builder in $\calG$ and
	a function $s_f: \PM(\calG) \to \bbS$, mapping partial $\calG$-matches to states 
	of $\bbS$, such that $\bbS,s_f(\calM) \sat_f \last(\calM)$ for every 
	$\overline{f}$-guided partial match $\calM \in \PM_P(\calG)$.
	
	The functions $\overline{f}$ and $s_f$ can be defined inductively. For the base case $|\calM| = 1$ the partial match $\calM$ consists of the single position $\Gamma$. We define $s_f(\calM) \isdef s$ and do not have to define $\overline{f}$ as this is a position owned by Prover.
	
	For the induction step we follow the specifications of the rule instance. If the rule is $\RuDia$, define $s_f$ as given by $f$ and let $\overline{f}$ choose the only premiss. For any other rule $s_f$ remains the same and we invoke Lemma \ref{lem.soundLocal} in order to choose a premiss for Builder.
	
	We need to show that $\overline{f}$ is a winning strategy for Builder in $\calG$.
	Because of Lemma \ref{lem.soundLocal} we know that all finite matches are won by
	Builder. 
	Thus, assume by contradiction that Builder loses an infinite 
	$\overline{f}$-guided $\calG$-match $\calM$. We write $\calM = \Delta_1i_1\Delta_2i_2...$ and let $\calM_n = \Delta_1i_1...i_{n-1}\Delta_n$ be the partial match up to position $\Delta_n$.
	Then cofinitely positions in $\calM$ of the form $\Delta_n$ have a formula in focus, where infinitely often Prover picks a rule instance of $\RuStarD$ where the principal formula is in focus. We will use $\calM$ to obtain an infinite $f$-guided $\calE$-match that is won by $\abel$. 
	Let $N \in \omega$ be such that $\Delta_n$ has a formula in focus for all $n \geq N$ and let $\psi_n^f \in \Delta_n$ be this formula in focus. 
	
	We claim that there is an $f$-guided $\calE@(\psi_N,s_f(\calM_N))$-match $\calP = P_1P_2...$ that is won by $\abel$ and such that $P_j = (\psi_n,s_f(\calM_n))$ for some $n \in \omega$. We define $P_1 = (\psi_N,s_f(\calM_N))$. Given $P_j = (\psi_n,s_f(\calM_n))$ let $k \geq n$ be minimal such that Prover picks a rule instance with principal formula $\psi_n^f$. This always exists as infinitely often Prover picks a rule instance of $\RuStarD$ where the principal formula is in focus. Then define $P_{j+1} = (\psi_{k+1}, s_f(\calM_{k+1}))$.
	The match $\calP$ is well-defined and $f$-guided by the definition of $s_f$ and $\overline{f}$. Moreover, $\psi_n$ is a diamond fixpoint formula infinitely often and thus $\calP$ is won by $\abel$. This implies $\bbS,s_f(\calM_N) \not\sat_f \psi_N$ contradicting $\bbS,s_f(\calM_N) \sat_f \Delta_N$.
\end{proof}

\begin{lemma}\label{lem.splitProofToProof}
	If $\proves \Gamma \| \Delta$, then $\proves \Gamma, \Delta$ for every split sequent $\Gamma \| \Delta$.
\end{lemma}
\begin{proof}
	Let $\pi$ be a $\SCPDLf$ proof of $\Gamma \| \Delta$. By translating every split rule in $\pi$ to its corresponding proof rule in Figure \ref{fig.rulesCPDL} we obtain a $\CPDLf$ proof of $\Gamma, \Delta$.
\end{proof}

\setcounter{theoremRem}{\thetheorem}
\setcounter{theorem}{\getrefnumber{thm.soundnessSplit}}
\addtocounter{theorem}{-1}
\begin{theorem}[Soundness]
	If $\proves \Gamma \| \Delta$, then $\Gamma, \Delta$ is unsatisfiable.
\end{theorem}
\begin{proof}
	Let $\pi$ be an \SCPDLf proof of $\Gamma \| \Delta$. Then Lemma \ref{lem.splitProofToProof} yields a \CPDLf proof of $\Gamma, \Delta$. Now Theorem \ref{thm.soundnessNormal} concludes the proof.
\end{proof}
\setcounter{theorem}{\thetheoremRem}

\subsection{Completeness}

We show the completeness of $\SCPDLf$ by contraposition: given a winning strategy for Builder in $\calG(\Sigma)@\Sigma$ we  find a pointed model $\bbS,s$ satisfying $\Sigma$. 

Let $f$ be a \emph{positional} winning strategy for Builder in $\calG(\Sigma)@\Sigma$.
We construct a pointed model $\bbS^f,s$ and a strategy $\underline{f}$ for $\eloi$ in $\calE(\bbS^f)$ such that $\bbS^f,s \sat_{\underline{f}} \Sigma$. We start by defining the pointed model $\bbS^f,s$.

Let $\calT$ be the subtree of the game-tree of $\calG(\Sigma)@\Sigma$, where Builder plays the strategy $f$ and Prover picks rule instances according to the following priorities: 
\begin{enumerate}
\item instances of \AxLit and \AxBot;
\item cumulative and productive instances of \RuAnd, \RuOr, \RuConD, \RuConB, \RuChoiceD, \RuChoiceB, \RuTestD, \RuTestB, \RuStarD, \RuStarB or \RuACut where the principal formula is in an unfocused component;
\item unblocking instances of \RuU;
\item cumulative and productive instances of \RuAnd, \RuOr, \RuConD, \RuConB, \RuChoiceD, \RuChoiceB, \RuTestD, \RuTestB, \RuStarD, \RuStarB or \RuACut where the principal formula is unfocused but in the focused component;
\item productive instances of \RuConD, \RuChoiceD, \RuTestD or \RuStarD, where the principal formula is in focus;
\item instances of \RuDia, cumulative instances of \RuF or conceding instances of \RuU.
\end{enumerate}

Additionally, at any two positions $\Sigma_0$ and $\Sigma_1$, where the focused components of $\Sigma_0$ and $\Sigma_1$ coincide and no rule instance of type 1 -- 3 is applicable, Prover picks the same rule instance of type 4, if possible.

Any winning strategy for Prover, where she picks rule instances according to those requirements, results in a uniform \SCPDLInfty proof. 

\begin{ourdefinition}
	We call a maximal path in $\calT$ of rule instances of type 1 -- 5 a \emph{local path}. 
	Let $\rho, \tau$ be local paths in $\calT$. We define $\rho \overset{a}{\to} \tau$ if $\tau$ is above $\rho$ in $\calT$, only separated by an instance of \RuDia[a] and possible instances of \RuF and conceding instances of \RuU.
	We let $\Plab^d(\rho) \isdef \bigcup\{\Sigma^d \| \Sigma \text{ occurs in } \rho\}$ for $d = l,r$ and $\Plab(\rho) \isdef \Plab^l(\rho) \cup \Plab^r(\rho)$.
\end{ourdefinition}

Let $\rho$ be a local path in $\calT$. Note that $\Plab^d(\rho)$ is not necessarily an annotated sequent as it may contain multiple formulas in focus. Because of the restriction on Prover's strategy $\Plab^d(\rho)$ is a \emph{saturated set}, meaning that no rule instance of type 1 -- 4 is applicable to $\Plab^d(\rho)^u$. Note that this definition conforms with the usual notion of a saturated set.
In particular, for every formula $\phi \in \FLN(\Plab^d(\rho)^-)$ it holds that either $\phi \in \Plab^d(\rho)^-$ or $\mybar{\phi} \in \Plab^d(\rho)^-$ and not both.

\begin{lemma}
	All local paths $\rho$ in $\calT$ are finite and $\Plab^d(\rho)^- = \Plab^d(\last(\rho))^-$.
\end{lemma}
\begin{proof}
	As there are only finitely many formulas in $\Phi$, and rule instances of type 2 and 4 are cumulative and productive, all paths consisting of rule instances of type 1 -- 4 are finite.  In rules of type 5 the principal formula is in focus. Therefore, if a local path $\tau$ is infinite it contains infinitely many rule instances of type 5 and no application of \RuU. This implies that $\tau$ is successful, which is a contradiction as we assumed that Builder plays a winning strategy.
	
	All rule instances $\begin{prooftree}
		\hypo{\Sigma_1}
		\hypo{\cdots}
		\hypo{\Sigma_n}
		\infer[left label = \Ru]3[]{\Sigma}
	\end{prooftree}$ of type 2 -- 5 are cumulative regarding the unannotated sequent, meaning that $\Sigma^-$ is a componentwise subset of $\Sigma_i^-$ for all $i = 1,..,n$.
	For all cumulative rule instances this is clear. The only non-cumulative ones are of type 5: rule instances where the principal formula is in focus. Yet we may assume that, if a rule instance with principal formula $\phi^f$ is chosen, then $\phi^u$ is in the sequent $\Sigma^d$ as well. Otherwise, either \AxLit or a cumulative and productive instance of cut with cut formula $\phi^u$ would be applicable depending on whether $\mybar{\phi}^u \in \Sigma^d$.	Hence, even though $\phi^f$ may not be in a premiss of the rule instance, $\phi^u$ is.
	Therefore, we inductively obtain $\Plab^d(\rho)^- = \Plab^d(\last(\rho))^-$.
\end{proof}

We can now define the model $\bbS^f = (S^f,R^f,V^f)$. We let $S^f$ be the set of local paths in $\calT$. We define $R^f = \{R_a^f\}_{a \in \Act}$ as follows:
\[\rho R_a \tau \quad :\Leftrightarrow \quad \rho \overset{a}{\to} \tau \text{ or } \tau \overset{\conv{a}}{\to} \rho\]
The valuation is defined as $V^f(p) \isdef \{\rho \in S^f \| p \in \Plab(\rho)^-\}$.

Next we define a strategy $\underline{f}$ for $\eloi$ in $\calE \isdef \calE(\bbS^f)$. This is done as follows:
\begin{enumerate}
	\item At $(\phi \lor \psi, \rho)$ pick $\phi$ if $\phi \in \Plab(\rho)^-$ and $\psi$ else.
	\item At $(\ldiap[a]\phi,\rho)$ choose some $\tau$ such that $\rho \overset{a}{\to} \tau$ by virtue of a rule instance of $\RuDia[a]$ with principal formula $\ldiap[a]\phi^f$ and as few applications of \RuF as possible.
	\item At $(\ldiap[\alpha \cup \beta]\phi,\rho)$ we make a case distinction:
	\begin{enumerate}
		\item If $\ldiap[\alpha] \phi$ is not in $\Plab(\rho)^-$, pick $\ldiap[\beta]\phi$,
		\item If $\ldiap[\beta] \phi$ is not in $\Plab(\rho)^-$, pick $\ldiap[\alpha]\phi$,
		\item Otherwise both $\ldiap[\alpha]\phi$ and $\ldiap[\beta]\phi$ are in $\Plab(\rho)^-$. If $\ldiap[\alpha \cup \beta] \phi \in \Plab^l(\rho)^-$ and at the rule instance \[\begin{prooftree}
			\hypo{\ldiap[\alpha]\phi^f, \Plab^l(\rho)^u \| \Plab^r(\rho)^u}
			\hypo{\ldiap[\beta]\phi^f,\Plab^l(\rho)^u \| \Plab^r(\rho)^u}
			\infer[left label=\RuChoiceD]2{\ldiap[\alpha \cup \beta] \phi^f,\Plab^l(\rho)^u \| \Plab^r(\rho)}
		\end{prooftree}\] 
		Builder chooses the left premiss, then pick $\ldiap[\alpha]\phi$ and if he chooses the right premiss pick $\ldiap[\beta]\phi$. 
		
		Else if at the rule instance 
		\[\begin{prooftree}
			\hypo{\Plab^l(\rho)^u \| \ldiap[\alpha]\phi^f, \Plab^r(\rho)^u}
			\hypo{\Plab^l(\rho)^u \| \ldiap[\beta]\phi^f,\Plab^r(\rho)^u}
			\infer[left label=\RuChoiceD]2{\Plab^l(\rho)^u \| \ldiap[\alpha \cup \beta] \phi^f, \Plab^r(\rho)^u}
		\end{prooftree}\]
		Builder chooses the left premiss, then pick $\ldiap[\alpha]\phi$ and if he chooses the right premiss pick $\ldiap[\beta]\phi$. 
	\end{enumerate}
	\item Analogously for $(\ldiap[\alpha^*]\phi,\rho)$.
	\item At $(\lboxp[\psi?]\phi,\rho)$ pick $\phi$ if it is in $\Plab(\rho)^-$ and $\mybar{\psi}$ else.
\end{enumerate}

\begin{remark}
    To give an intuitive explanation of the somewhat exotic definition of $\underline{f}$ we have to think about our proof strategy.
    The aim of this definition is to ensure that in an $\underline{f}$-guided $\calE(\bbS^f)$-match $\calM = (\psi_n, \rho_n)_{n<\kappa}$ it holds $\psi_n \in \Plab(\rho_n)^-$. This already guarantees that all such finite matches are won by $\eloi$.

    For infinite matches we also have to take the annotations into account. 
    In order to show that all infinite matches are won by $\eloi$ we argue by contraposition. Given an $\underline{f}$-guided infinite $\calE$-match $\calM= (\psi_n, \rho_n)_{n\in \omega}$ that is won by $\abel$ we have to find an infinite successful path in $\calT$. Because $\calM$ is won by $\abel$ there is $N \in \omega$ such that $\psi_n$ is a diamond formula for all $n \geq N$. We aim to find a path in $\calT$ where $\psi_n$ is in focus in $\Plab(\rho_n)$ for all $n \geq N$. Yet we have to be very careful, $\psi_n$ might be in focus in the left or the right component of $\Plab(\rho_n)$ and Builder's strategy $f$ might differ in the cases where $\psi_n^f$ is principal in the left or the right component. This explains our complicated definition of $\underline{f}$ in case 3(c), in which we give priority to the left component: If $\psi_M \in \Plab^l(\rho_M)$ for some $M \geq N$ (guaranteeing that $\psi_n \in \Plab^l(\rho_n)$ for all $n \geq M$) we may then find a path in $\calT$ where $\psi_n$ is in focus in the left component for all $n \geq N$. If on the other hand $\psi_n \notin \Plab^l(\rho_n)$ for all $n \geq N$ (guaranteeing that $\psi_n \in \Plab^r(\rho_n)$ for all $n \geq N$) we find a path in $\calT$ where $\psi_n$ is in focus in the right component for all $n \geq N$.
Note that this definition is only possible because we assume that the strategy $f$ for Builder is positional. 

\end{remark}

In order to show that the strategy $\underline{f}$ is well-defined we have to guarantee that at any position $(\ldiap[a]\phi,\rho)$ in a match it holds that $\ldiap[a]\phi \in \Plab(\rho)^-$. Then, at $\last(\rho)$ in $\calT$ Prover might put $\ldiap[a]\phi$ in focus and apply an $\RuDia[a]$ rule. For matches starting at the root of $\calT$ the next lemma guarantees that for any positions $(\ldiap[a]\phi,\rho)$ indeed it holds $\ldiap[a]\phi \in \Plab(\rho)^-$.

\begin{lemma}[Truth Lemma]\label{lem.truthLemma}
	Let $\psi_0 \in \Sigma^d$ and let $\rho_0$ be a local path containing $\Sigma$. Let $\calM = (\psi_n,\rho_n)_{n < \kappa}$ be an $\underline{f}$-guided $\calE$-match starting at $(\psi_0,\rho_0)$. Then for every $n < \kappa$ it holds that $\psi_n \in \Plab^d(\rho_n)^-$.
\end{lemma} 
\begin{proof}
	We prove the claim by strong induction on $n$. The base case holds by assumption. For the induction step let $\psi_n \in \Plab^d(\rho_n)^-$, we need to show that $\psi_{n+1} \in \Plab^d(\rho_{n+1})^-$. We proceed by a case distinction on the shape of $\psi_n$. If $\psi_n$ is not of the form $\lboxp[a]\chi$ or $\ldiap[a]\chi$ for some action $a$, then $\rho_{n+1} = \rho_n$ and the claim easily follows from the fact that $\Plab^d(\rho_n)$ is a saturated set and the definition of $\underline{f}$.
	
	Assume $\psi_n = \lboxp[a]\chi$, then $\psi_{n+1} = \chi$. In this case either $\rho_n \overset{a}{\to} \rho_{n+1}$ or $\rho_{n+1} \overset{\conv{a}}{\to} \rho_{n}$. First assume that $\rho_n \overset{a}{\to} \rho_{n+1}$. Then $\lboxp[a]\chi^u$ is in the conclusion of the rule instance of \RuDia[a] between $\rho_n$ and $\rho_{n+1}$, hence $\chi^u$ is in its premiss and therefore $\chi \in \Plab^d(\rho_{n+1})^-$.
	
	Now consider the case $\rho_{n+1} \overset{\conv{a}}{\to} \rho_{n}$. We first show that $\lboxp[a] \chi \in \FLN(\Plab^d(\rho_{n+1})^-)$. As $\bbS^f$ is a forest and $\rho_0...\rho_{n+1}$ forms a path in $\bbS^f$ starting at one of the roots, where the last step of the path is downwards, there has to be an $i \in \{ 0,...,n-1 \}$ with $\rho_i = \rho_{n+1}$. In the match $\calM$ there are positions $(\psi_i,\rho_i)$ and $(\lboxp[a]\chi, \rho_n)$, hence $\lboxp[a]\chi \in \FL(\psi_i)$. By induction hypothesis $\psi_i \in \Plab^d(\rho_i)^- = \Plab^d(\rho_{n+1})^-$ and thus $\lboxp[a]\chi \in \FL(\Plab^d(\rho_{n+1})^-)$.
	
	Towards a contradiction assume that $\chi \notin \Plab^d(\rho_{n+1})^-$. 
	Then, because $\Plab^d(\rho_{n+1})$ is a saturated set and $\chi \in \FLN(\Plab^d(\rho_{n+1})^-)$ it holds that $\mybar{\chi} \in \Plab^d(\rho_{n+1})^-$. Let $I$ be the rule instance of \RuDia[\conv{a}] between $\rho_{n+1}$ and $\rho_n$. The formula $\mybar{\chi}^u$ is in the conclusion of $I$, therefore $\ldiap[\conv{\conv{a}}]\mybar{\chi}^u = \ldiap[a] \mybar{\chi}^u$ is in its premiss as $\ldiap[a] \mybar{\chi} \in \FLN(\Plab^d(\rho_{n+1})^-)$. This implies $\ldiap[a] \mybar{\chi} \in \Plab^d(\rho_{n})^-$, contradicting the fact that $\Plab^d(\rho_{n})$ is a saturated set and $\lboxp[a] \chi \in \Plab^d(\rho_{n})^-$.
	
	Lastly assume that $\psi_n = \ldiap[a]\chi$ and $\psi_n \in \Plab^d(\rho_n)^-$. Then $\psi_n^o \in \last(\rho_n)^d$ and by the definition of the strategy $\underline{f}$ it holds that $\rho_n \tom{a} \rho_{n+1}$. If $\psi_n^f \in \Plab^d(\rho_n)$, this means in $\calT$ a rule instance of $\RuDia^d$ with principal formula $\rho_n^f$ is applied and therefore $\psi_{n+1}^f \in \Plab^d(\rho_{n+1})$. Otherwise Prover may first put $\psi_n$ in focus and then apply the rule instance $\RuDia^d$, yielding the same sequent.
\end{proof}

\begin{lemma}\label{lem.truthLemmaDiamonds}
Let $\Psi$ be a split sequent. Let $\psi_0^f \in \Psi^d$ and let $\rho_0$ be a local path containing $\Psi$. Let $\calM = (\psi_n,\rho_n)_{n \in \omega}$ be an $\underline{f}$-guided $\calE$-match starting at $(\psi_0,\rho_0)$ such that for all $n \in \omega$ it holds that $\psi_n$ is a diamond formula, and if $\psi_n = \ldiap[\phi?]\chi$ then $\psi_{n+1} = \chi$.
	
If either $d = l$ or $d = r$ and $\psi_n \notin \Plab^l(\rho_n)^-$ for all $n \in \omega$, then $\psi_n^f \in \Plab^d(\rho_n)$ for all $n \in \omega$. If additionally $\psi_n$ is of the form $\psi_n = \ldiap[a] \chi$, then $\psi_n^f \in \last(\rho_n)^d$.	
\end{lemma}
\begin{proof}
We prove the claim by induction on $n$. The base case holds by assumption. For the induction step let $\psi_n^f \in \Plab^d(\rho_n)$. We need to show that $\psi_{n+1}^f \in \Plab^d(\rho_{n+1})$. We proceed by a case distinction on the shape of $\psi_n$. If $\psi_n$ is not of the form $\ldiap[a]\chi$ for some action $a$, then $\rho_{n+1} = \rho_n$.

	The cases $\psi_n = \ldiap[\alpha \cup \beta] \chi$ and $\psi_n = \ldiap[\alpha^*]\chi$ follow by the definition of the strategy $\underline{f}$. In the case $d=r$ we need the extra condition that $\psi_n \notin \Plab^l(\rho_n)^-$ for all $n \in \omega$ to guarantee that the choice of $\eloi$ coincides with the choice of Builder in the proof search game.
	The case where $\psi_n = \ldiap[\alpha;\beta]\chi$ is clear and the case $\psi_n = \ldiap[\phi?]\chi$ follows by assumption. 
	
	Lastly assume that $\psi_n = \ldiap[a]\chi$ and $\psi_n^f \in \Plab^d(\rho_n)$. Then $\psi_n^f \in \last(\rho_n)^d$ and by the definition of the strategy $\underline{f}$ it holds that $\rho_n \tom{a} \rho_{n+1}$. Hence in $\calT$ Prover applies a rule instance of $\RuDia^d$ with principal formula $\psi_n^f$ and therefore $\psi_{n+1}^f \in \Plab^d(\rho_{n+1})$.
\end{proof}

\begin{lemma}\label{lem.completenessWinning}
	Let $\psi_0 \in \Sigma$ and let $\rho_0$ be a local path containing $\Sigma$. Then the strategy $\underline{f}$ is winning for $\eloi$ in $\calE@(\psi_0,\rho_0)$.
\end{lemma}
\begin{proof}
	Let $\calM$ be an $\underline{f}$-guided $\calE@(\psi_0,\rho_0)$-match. If $\calM$ is finite it is straightforward to check that it is winning for $\eloi$. Thus we consider the case where $\calM = (\psi_n,\rho_n)_{n \in \omega}$ is infinite and assume that it is winning for $\abel$. Then there is $N \in \omega$ such that $\psi_n$ is a diamond formula for all $n \geq N$.  Without loss of generality we may assume that $N$ is big enough such that for all formulas $\psi_n$ of the form $\ldiap[\phi?]\chi$ (with $n \geq N$) it holds that $\psi_{n+1} = \chi$. 
	
Let $M \geq N$ be such that $\psi_M \in \Plab^l(\rho_M)^-$, or if such an $M$ does not exist (meaning that $\psi_n \notin \Plab^l(\rho_n)$ for all $n \geq N$), let $M \isdef N$. In the first case let $d\isdef l$ and in the second let $d \isdef r$. In both cases it holds that $\psi_n \in \Plab(\rho_n)^-$ for all $n \in \omega$ by Lemma \ref{lem.truthLemma}.
	\medskip
	
	We first assume that for infinitely many $n$ the formula $\psi_n$ is of the form $\ldiap[a_n]\chi$ for some (atomic) program $a_n$. Let $K \geq M$ such that $\psi_K$ is of the form $\ldiap[a]\chi$, then $\psi_{K+1} = \chi$ and $\chi^f \in \Plab^d(\rho_{K+1})$.
	By Lemma \ref{lem.truthLemmaDiamonds} we have $\psi_n^f \in \Plab^d(\rho_n)$ for all $n > K$. Additionally, for every $n >  K$ there is some $m \geq n$ such that $\psi_m = \ldiap[a_m] \chi$ and $\psi_m^f \in \last(\rho_n)^d$. But if the last sequent on $\rho_n$ has a formula in focus, then all sequents in $\rho_n$ have a formula in focus for all $n > K$, as no \RuF rule is applied on local paths. Between the local paths $\rho_n$ and $\rho_{n+1}$ with $\rho_n \neq \rho_{n+1}$ no \RuF rule is applied either, as by the definition of $\underline{f}$ applications of \RuF are minimized. Thus there is an infinite path in $\calT$ where cofinitely many sequents have a formula in focus and infinitely many rule instances with principal formula in focus are applied. This contradicts the assumption that Builder plays a winning strategy.
	\medskip
	
Now assume that for some $K \geq M$ there is no $n \geq K$ such that the formula $\psi_n$ is of the form $\ldiap[a]\chi$. As a consequence $\rho_n = \rho_N$ for all $n \geq K$. 
	
Let $\last(\rho_K) = \Pi \| \Xi$, we show how to obtain an infinite successful path in $\calT$ starting at $\Pi \| \Xi$ contradicting the assumption that Builder plays a winning strategy. 
Assume that $d = l$, the case where $d = r$ is analogous. 
At $\Pi \| \Xi$ no rule instance of type 1 -- 5 is applicable, thus in $\calT$ Prover may put $\psi_K$ in focus to obtain a node $\psi_K^f, \Pi^u \| \Xi^u$. We will show that there is an infinite path $\Sigma_0 i_0 \Sigma_1 i_1 ...$ in $\calT$, where for all $n \in \omega$ it holds $\Sigma_n = \psi_{K+n}^f, \Pi^u \| \Xi^u$ and $i_n$ is a rule instance of type 5. 
	
Inductively assume that the position $\psi_{K+n}^f, \Pi^u \| \Xi^u$ is in $\calT$. Following her strategy in $\calT$ Prover may pick a rule instance of type 5 with principal formula $\psi_{K+n}^f$. We proceed by a case distinction on the shape of $\psi_{K+n}$. 
	
If $\psi_{K+n} = \ldiap[\alpha \cup \beta] \chi$, then $\ldiap[\alpha \cup \beta] \chi \in \Plab^l(\rho_N)^-$ and thus $\eloi$ picks the correct formula according to Builder's strategy $\underline{f}$.
Analogously for $\psi_{K+n} = \ldiap[\alpha^*]\chi$.
The case where $\psi_{K+n} = \ldiap[\alpha;\beta]\chi$ is clear. 
If $\psi_{K+n} = \ldiap[\phi?]\chi$, then by the definition of $N$, and the fact that $K \geq N$, it follows
$\psi_{K+n+1} = \chi$. The rule instance $i_n$ only has one premiss $\phi^u, \chi^f, \Pi^u \| \Xi^u$. As $\Pi$ is saturated either $\phi$ or $\mybar{\phi}$ is in $\Pi$. If $\phi \in \Pi$ we have shown the induction step, if on the other hand $\mybar{\phi} \in \Pi$ then an axiom would be applicable, contradicting the fact that Builder's strategy is winning. 
	
Thus there is an infinite path in $\calT$ where cofinitely many sequents have a formula in focus and infinitely often a rule instance with principal formula in focus is applied, contradicting that Builder's strategy $f$ is winning. 
\end{proof}

\setcounter{theoremRem}{\thetheorem}
\setcounter{theorem}{\getrefnumber{thm.completenessSplit}}
\addtocounter{theorem}{-1}
\begin{theorem}[Completeness]
	If $\Gamma,\Delta$ is unsatisfiable then there is a uniform \SCPDLf proof of $\Gamma \| \Delta$.
\end{theorem}
\begin{proof}
	We assumed that $f$ is a winning strategy for Builder in $\calG(\Gamma\|\Delta)@(\Gamma\|\Delta)$. Let $\rho_0$ be a local path in $\calT$ containing $\Gamma \| \Delta$. Then Lemma \ref{lem.completenessWinning} shows that $\bbS^f, \rho_0 \sat_{\underline{f}} \Gamma, \Delta$, which implies that $\Gamma, \Delta$ is satisfiable. By contraposition this means that for every unsatisfiable sequent $\Gamma, \Delta$ Prover has a winning strategy in $\calG(\Gamma\|\Delta)@(\Gamma\|\Delta)$ and thus there is a $\SCPDLInfty$ proof $\pi$ of $\Gamma \| \Delta$, which may be assumed to be regular. By our restriction on the strategy of Prover $\pi$ is actually a uniform proof.
    
Using Lemma \ref{lem.CPDLfIffCPDLInfty} we obtain an \SCPDLf proof $\rho$ of $\Gamma\|\Delta$. By a quick inspection on the proof of Lemma \ref{lem.CPDLfIffCPDLInfty} it follows that $\rho$ is uniform because $\pi$ is uniform.
This concludes the proof.
\end{proof}
\setcounter{theorem}{\thetheoremRem}

\end{document}